\documentclass[10pt,journal]{IEEEtran}
\usepackage{amsmath,amssymb,amsfonts,bm}
\usepackage{graphicx}
\usepackage{booktabs}
\usepackage{array}
\usepackage{cite}
\usepackage{xcolor}
\usepackage{algpseudocode}
\usepackage{url}
\usepackage{mathrsfs}
\usepackage[hidelinks]{hyperref}

\newtheorem{theorem}{Theorem}

\newtheorem{proposition}{Proposition}
\newtheorem{cor}{Corollary}

\newtheorem{rmk}{Remark}

\renewcommand{\le}{\leqslant}
\renewcommand{\leq}{\leqslant}

\renewcommand{\geq}{\geqslant}

\newcounter{algorithm}

\newcommand{\algorithmheading}[1]{%
  \refstepcounter{algorithm}%
  \hrule height 0.8pt\relax\kern 2pt
  \noindent\textbf{Algorithm~\thealgorithm.}\enspace #1\par
  \kern 2pt\hrule height 0.4pt\relax\kern 2pt}
\newcommand{\algorithmendrule}{%
  \par\kern 2pt\hrule height 0.8pt\relax}

\allowdisplaybreaks

\title{Prediction-Aware Structured Resource Control for Partitioned IRS-Assisted Mobile IoT Uplinks With Semi-Blind Cascaded-Channel Acquisition}

\author{Hediyeh~Soltanizadeh and~Ardavan~Rahimian,~\IEEEmembership{Senior Member,~IEEE}%
\thanks{\textit{(Corresponding author: Ardavan Rahimian.)}}%
\thanks{H.~Soltanizadeh is with the School of Electrical Engineering, Iran University of Science and Technology (IUST), Tehran, Iran
(e-mail: he.soltanizadeh@gmail.com).}%
\thanks{A.~Rahimian is with the School of Engineering, Ulster University, Belfast BT15 1AP, U.K.
(e-mail: a.rahimian@ulster.ac.uk).}}

\begin{document}

\maketitle

\begin{abstract}
Intelligent reflecting surface (IRS)-assisted mobile Internet of Things (IoT) uplinks require joint control of channel aging, costly cascaded channel state information (CSI) acquisition, and coupled IRS/radio resources. This investigation develops a unified prediction-aware deterministic framework for partitioned IRS-assisted massive multiple-input multiple-output (mMIMO) uplinks. Direct channels are recursively tracked, whereas cascaded channels are selectively refreshed by differential semi-blind acquisition reusing the same unknown data block across a minimum-reflection baseline and discrete Fourier transform (DFT)-coded IRS states. Moreover, finite-block acquisition covariances initialize reduced-order beam-domain prediction, which propagates only unit-beam means online and retrieves age-dependent covariances from precomputed tables. A mixed-integer nonconvex formulation captures net throughput, fairness, outage, uncertainty, switching cost, and sliding-window IRS service. The causal controller integrates covariance-floor actionability, feasible equal-size granularity, residual-feasible ownership, isolated-beam one-sweep refinement, same-block fresh-payload accounting, and service-first allocation using recomputed uncertainty-aware zero-forcing (ZF) rates and bidirectional receiver-consistent risk. Analysis establishes service-feasibility preservation, effective-channel error bounds, distribution-free reliability, and polynomial online scaling. In unseen environments, the predicted acquisition floor matches the practical semi-blind error, actionable recalibration reduces the mean recalibration rate by \(31.5\%\), and IRS-assisted control achieves a \(14.09\%\) mean net-rate gain over direct-only transmission. Receiver-consistent allocation improves realized rate, while ablations support equal-size partitioning and single-sweep refinement; zero post-action service violations occur under nominal, mobility, and load-stress conditions.
\end{abstract}


\section{Introduction}
The increasing deployment of mobile and vehicular Internet of Things (IoT) applications requires wireless networks that can efficiently adapt to user mobility, dynamic propagation environments, and limited radio resources, while maintaining reliable connectivity \cite{ TAP24, CHXTHZ24}. Massive multiple-input multiple-output (mMIMO) is a promising enabling technology for meeting these requirements through spatial multiplexing, coherent beamforming, and effective interference management \cite{ CJZL21}. Intelligent reflecting surfaces (IRSs) offer a complementary solution by using a large number of nearly passive reflecting elements to intelligently redirect signals and improve coverage without requiring an individual radio frequency (RF) chain for each element. Integrating IRSs with mMIMO enables adaptive control of the wireless propagation environment, providing additional spatial degrees of freedom for more efficient interference management and enhanced link reliability \cite{GV22,AKOZ25}.

Nevertheless, the performance gains offered by mMIMO-assisted IRSs rely heavily on accurate and timely channel state information (CSI), which remains difficult to maintain in mobile IoT scenarios due to channel aging, pilot limitations, and the large dimensionality of cascaded IRS channels \cite{IMDGB25}. As channel estimates rapidly become outdated in mobile or vehicular scenarios, resource-allocation decisions based on stale CSI may experience significant performance degradation. This difficulty is more pronounced in IRS-assisted systems since a passive IRS has no dedicated RF chains for transmitting pilots or processing baseband observations. As a result, the dimension of the cascaded channel grows with the number of reflecting elements, which can lead to substantial training overhead. This issue becomes more severe in mobile or dense environments, where rapid channel variations, channel aging, topology changes, and pilot scarcity require mobility-aware tracking methods that recursively combine previous CSI with current observations for both active and inactive devices \cite{WLZC25, SF24, SF21, SF22}. Existing CSI-acquisition methods are mainly designed for static or slowly varying settings and become increasingly inefficient when user mobility, selective cascaded-CSI refresh, and multiuser IRS sharing must be handled jointly.

Recent IRS research has increasingly addressed the CSI overhead and channel aging that arise when cascaded channels must be maintained under mobility. Individual IRS-channel acquisition with explicit accuracy-overhead tradeoffs was studied in \cite{Zhou2024IndividualCE}, while a two-timescale estimator exploiting the slower variation of the IRS-BS link was developed in \cite{Huang2024TwoTimescaleCE}. Multiuser cascaded-channel structure was used in \cite{Wang2025QuantizeEstimate} to reduce estimation and feedback overhead, whereas reduced-dimensional estimation and channel prediction under fast fading and channel aging were combined in \cite{Ginige2025ChannelPrediction}. Two-timescale massive-MIMO operation under imperfect CSI was further investigated in \cite{Zhi2023TwoTimescale}. More recently, CSI acquisition, feedback, and active/passive beamforming were jointly optimized in an end-to-end IRS framework in \cite{Cui2026E2ECSI}, bringing channel acquisition closer to the subsequent communication design. Robust IRS-assisted transmission under high mobility and imperfect cascaded CSI was considered in \cite{Liu2026RobustHST} using bounded and statistical CSI-error models. These studies motivate structured temporal CSI processing and uncertainty-aware design, but they do not jointly address selective cascaded-CSI refresh, reduced-order posterior propagation, and causal multiuser resource control.

Semi-blind and data-aided methods provide a means of reducing training overhead. Pilot and data symbols were jointly exploited for semi-blind reflected-channel estimation in \cite{AE22}. Data-aided recovery of the combined channel and transmitted symbols was considered in \cite{Magalhaes2025DataAided}, while joint semi-blind channel estimation and symbol detection for multi-IRS MIMO was studied in \cite{Li2025SemiBlind}. These methods mainly target channel-recovery accuracy and pilot reduction. They do not explicitly connect finite-block acquisition uncertainty to reduced-order prediction, selective CSI refresh, and the payload cost incurred by recalibration within the same physical block.

Resource control for IRS-assisted multiuser systems was studied from several dimensions, including surface partitioning, imperfect CSI, fairness, and long-term service constraints. Joint IRS partitioning and multiuser scheduling were studied in
\cite{Hashida2026Partitioning}, whereas \cite{Wu2025RobustRA} considered discrete-phase resource allocation under both perfect and imperfect CSI. Long-term channel assignment and IRS configuration under queue-stability requirements were addressed in \cite{Jia2025DelayAware}, while \cite{Zivuku2025Fairness} developed a robust multi-IRS design accounting for geographical fairness, service availability, and worst-case quality of service (QoS) under bounded CSI uncertainty. More recently, \cite{Chan2026LoadBalanced} jointly optimized load-balanced user association, BS beamforming, and IRS phase control in multi-BS mmWave networks using uplink pilots directly, thus bypassing explicit CSI reconstruction. Although these studies cover important aspects of partitioning, robust allocation, fairness, queue stability, and load balancing, they do not jointly couple selective cascaded-CSI refresh and its same-block payload cost with reduced-order uncertainty propagation, service-feasible IRS sharing, beam refinement, and receiver-consistent cochannel allocation.

The closest system-level baseline is the mobility-aware IRS-assisted uplink framework in \cite{ArdavanBaseline}, which combines finite-resolution IRS focusing, inverse-rate user prioritization, and sequential energy-detection-based channel allocation. However, its IRS control is restricted to single-user focusing, its CSI evolution is simplified, and its uplink-channel assignment relies on a sensing-based heuristic. Motivated by these limitations, this work develops a unified prediction-aware deterministic framework for partitioned IRS-assisted mobile mMIMO IoT uplinks, integrating selective cascaded-CSI refresh, service-feasible multiuser IRS sharing, and receiver-consistent uplink allocation. The contributions are as follows:

\begin{enumerate}

\item A mobility-aware acquisition combines recursive direct-channel tracking with differential semi-blind cascaded-channel recalibration. Reusing the same unknown data block under a minimum-reflection baseline and DFT-coded IRS states removes the direct component before recovery. Anchor-mismatch analysis explains the resulting cross-user error floor, while offline finite-block covariances initialize beam-domain posterior uncertainty.

\item A reduced-order beam-domain predictor is developed over fixed elementary IRS units. Only \(M\)-dimensional unit-beam means are propagated online, while age-dependent error covariances are retrieved from precomputed tables to form uncertainty-aware effective channels and zero-forcing (ZF) rates, avoiding full \(MN\)-dimensional covariance propagation. The uncertainty model is validated against empirical prediction errors.

\item A mixed-integer nonconvex formulation captures net throughput, fairness, outage, prediction uncertainty, IRS switching cost, and sliding-window service. Its deterministic realization combines covariance-floor actionability, feasible equal-size granularity, residual-feasible ownership, and isolated-beam initialization with one coordinate sweep. Ablations support the equal-size and single-sweep choices over more complex alternatives.

\item The full-slot controller charges recalibration overhead through same-block fresh-payload accounting and preserves sliding-window IRS service. Uplink channels are assigned in service-first order using recomputed uncertainty-aware ZF rates and bidirectional receiver-consistent risk, improving performance over rate-only allocation without exhaustive predicted-rate selection.

\item Theoretical studies establish IRS-service feasibility and sequential preservation, effective-channel prediction-error bounds, and a distribution-free reliability bound. Complexity analysis shows polynomial online scaling. On unseen environments, actionable recalibration reduces the mean recalibration rate by \(31.5\%\), IRS-assisted control achieves a \(14.09\%\) mean net-rate gain over direct-only transmission, and no post-action service violations occur under nominal, mobility, or load stress.

\end{enumerate}

\section{System Model}

Consider a single-cell  mMIMO system in which an $M$-antenna base station (BS) serves $Z$ single-antenna mobile IoT users under time-division duplexing. Since the number of users exceeds the available orthogonal pilot sequences $\tau$, each pilot sequence is assigned to a group of $K=Z/\tau$ users, assuming that $K$ is an integer.
Because users assigned to different pilot groups can be separated through pilot orthogonality, the analysis can be restricted to the first group IoT users $\mathcal K=\{1,2,\ldots,K\}$ assigned to the pilot sequence $\boldsymbol\phi_1\in\mathbb C^{\tau_p\times 1}$.  Although the formulation is applicable to a MIMO-OFDM system over one coherence band, it can equivalently represent a narrowband single-carrier system.

Each user independently decides whether to transmit its assigned pilot during each time slot. 
Let $b_t^{(k)}\in\{0,1\}$ denote the pilot-activity indicator of user $k$, where \(b_t^{(k)}=1\) if the pilot is transmitted in slot \(t\) and \(b_t^{(k)}=0\) otherwise. Define the set of pilot-active users $\mathcal D_t \triangleq \big\{ k\in\mathcal K: b_t^{(k)}=1 \big\}$. The BS identifies $\mathcal D_t$ through the coordinated random-access protocol before channel estimation, and the scheduled data-user set is assumed to coincide with $\mathcal D_t$. Since this paper considers a high-rate IoT scenario, \(\mathcal D_t \neq\emptyset\) is assumed throughout.

\begin{figure}[t!]
\centering
\includegraphics[width=\linewidth]{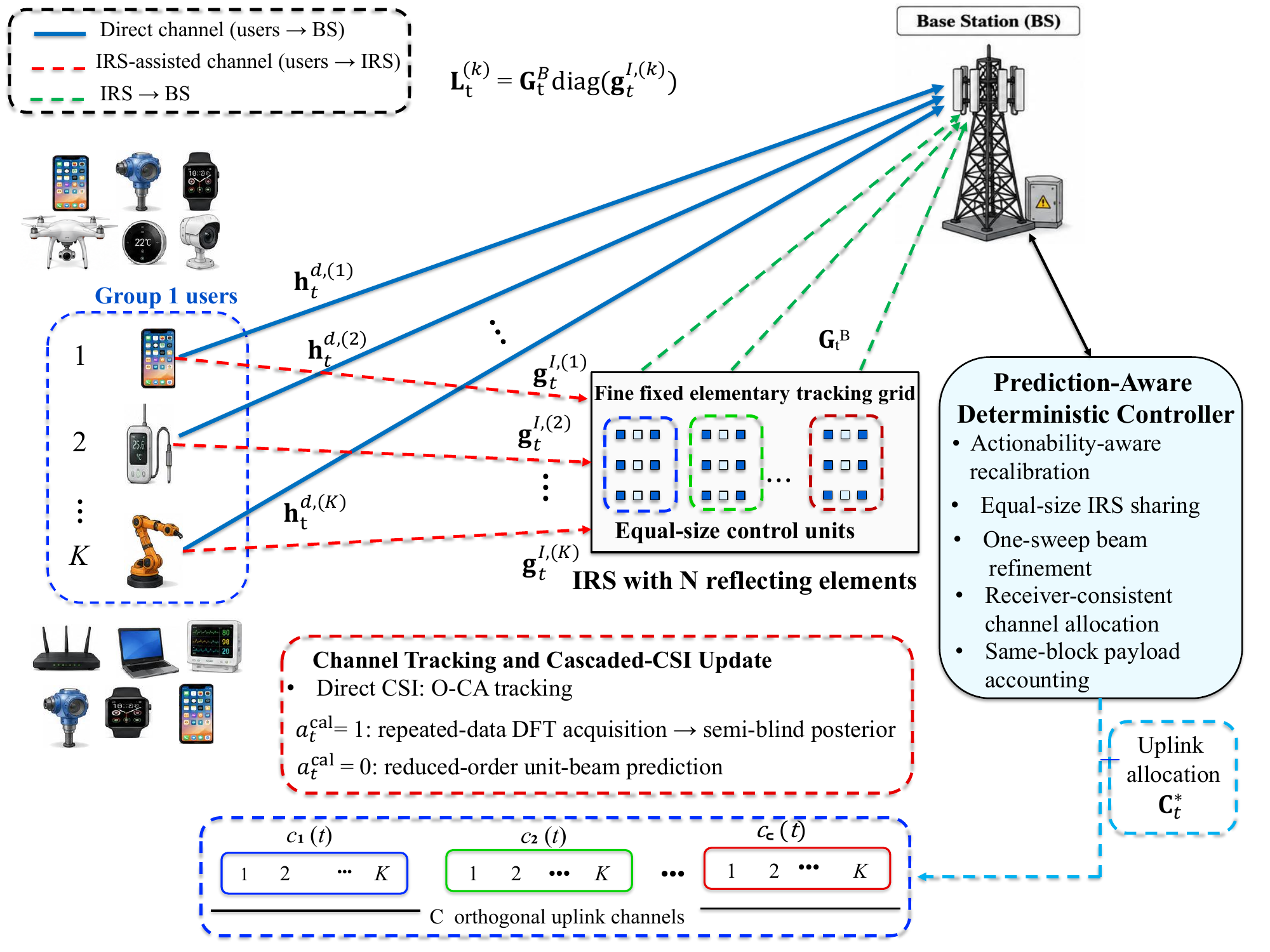}
\caption{System model and causal prediction-aware control architecture for
the partitioned IRS-assisted mobile mMIMO IoT uplink.}
\label{fig:system_model}
\end{figure}

The system also includes a fully passive IRS with $N$ reflecting elements. Let $\mathcal N =\{1,\ldots,N\}$ and $\mathcal C=\{1,\ldots,C\}$ denote the IRS-element set and the set of orthogonal uplink channels, respectively. The uplink channel assigned to user $k$ at slot $t$ is denoted by $ c^{(k)}_{t}$.  Let $\mathbf h_{t}^{d, (k)}\in\mathbb C^{M \times 1}$ denote the direct-channel from user $k$ to the BS, $\mathbf g_{t}^{I, (k)}\in\mathbb C^{N \times 1}$ the user-to-IRS channel, and $\mathbf G_{t}^{B}\in\mathbb C^{M\times N}$ the IRS-to-BS channel, for $k=1,\ldots,K$  (see Fig.~\ref{fig:system_model}).  The IRS-to-BS channel is common to all users, whereas the direct and user-to-IRS channels are user dependent.  Because the IRS is fully passive, the BS observes the user-to-IRS and IRS-to-BS links through their cascaded response rather than as separately identifiable channels. For user $k$, the cascaded channel across the $N$ IRS elements is defined as
\begin{equation}
\mathbf L_{t}^{(k)} \triangleq \mathbf G_{t}^{B} \operatorname{diag}\!\big(
\mathbf g_{t}^{I, (k)} \big) \in\mathbb C^{M\times N},
\label{eq:1}
\end{equation}
where $\mathbf L_{t}^{(k)}=\big[\bm{l}_{1,t}^{(k)},\ldots,\bm{l}_{N,t}^{(k)}\big]$ and
$\bm{l}_{n,t}^{(k)}=[\mathbf G_t^{B}]_{:,n}[\mathbf g_t^{I,(k)}]_n\in\mathbb C^{M \times 1}$. The proposed semi-blind acquisition framework directly estimates $\mathbf L_{t}^{(k)}$, rather than separately estimating $\mathbf G_{t}^{B}$ and $\mathbf g_{t}^{I, (k)}$.

The direct channel is $\mathbf h_t^{d,(k)}\sim \mathcal{CN} \big(\mathbf 0,d_t^{d,(k)}\mathbf I_M\big),$ where \(d_t^{d,(k)}\) denotes the slow-timescale large-scale attenuation of the direct link, accounting for path-loss variation and shadowing. Following the conventional mMIMO assumptions, the direct user channels are assumed to exhibit favorable propagation over the BS array. Moreover, the direct channels are asymptotically uncorrelated with the cascaded components. For \(i\neq k\), the direct user channels satisfy:
\begin{align}
&\frac{1}{M}(\mathbf h_t^{d,(i)})^{H} \mathbf h_t^{d,(k)} \xrightarrow[M\rightarrow\infty]{\mathrm{a.s.}}0, \label{eq:2}\\
& \frac{1}{M}(\mathbf h_t^{d,(k)})^{H} \mathbf h_t^{d,(k)}\xrightarrow[M\rightarrow\infty]{\mathrm{a.s.}} d_t^{d,(k)}. \label{eq:3}
\end{align}
Moreover, for all $i$, $k$, and $n\in\mathcal N$,
\begin{align}
\frac{1}{M} (\bm{l}_{n,t}^{(i)})^{H} \mathbf h_t^{d,(k)} \xrightarrow[M\rightarrow\infty]{\mathrm{a.s.}}0. \label{eq:4}
\end{align}
These properties support the use of the tracked direct-channel matrix as a multiuser anchor in Section~IV-A, where the baseline and differential observations are used to separate the user-specific cascaded responses.

To model the temporal evolution of the direct channel induced by users' mobility, we adopt the first-order Gauss-Markov model in \cite{SFR21}:
\begin{equation}
\mathbf{h}_t^{d, (k)}=\mathbf{A}_t^{d, (k)}\mathbf{h}_{t-1}^{d, (k)}+\mathbf{u}_t^{d, (k)}, \label{eq:5}
\end{equation}
where $\mathbf{A}_t^{d, (k)}$ is the temporal correlation matrix, and $\mathbf{u}_t^{d, (k)}\sim \mathcal{CN}(\mathbf{0},\mathbf{Q}_t^{d, (k)})$ is the channel state noise to model small-scale fading. The matrices $\mathbf A_t^{d,(k)}$ and $\mathbf Q_t^{d,(k)}$ are obtained from long-term channel statistics and are selected consistently with the marginal covariance $d_t^{d,(k)}\mathbf I_M$. Let $\mathbf Y_t$ denote the received pilot matrix at the BS.  In order to suppress the contributions of users belonging to the other groups, the received pilot matrix is projected onto the orthonormal pilot sequence assigned to group 1:
\begin{equation}
\mathbf y_t=\mathbf Y_t\boldsymbol\phi_1=\sum_{k\in\mathcal D_t}\mathbf h_t^{d,(k)}
+\mathbf v_t,   \label{eq:6}
\end{equation}
where \(\mathbf v_t=\mathbf V_t\boldsymbol\phi_1\sim\mathcal{CN}(\mathbf 0,\mathbf R_t)\).  Using all pilot observations collected up to slot $t$, the BS updates the channels of active users while predicting the channels of inactive users.

The direct and cascaded channels are treated separately because they follow different observation models. Section~III tracks the direct channels from pilot observations using the optimal coordinated-access (O-CA) algorithm. Section~IV then acquires the cascaded channels from repeated unknown payload blocks collected under a minimum-reflection baseline and DFT-coded IRS states. Hence, baseline subtraction removes the direct component directly at the observation level before cascaded-channel recovery. Recalibration is performed only on selected slots, while the cascaded states are propagated by prediction between recalibration events.
\begin{rmk}
For a scalar metric $x_t$, the clipped training-statistics normalization is defined as
\begin{equation}
\mathcal N_x(x_t) \triangleq \left[\frac{x_t-x_{\min}^{\mathrm{tr}}} {x_{\max}^{\mathrm{tr}}-x_{\min}^{\mathrm{tr}}+\epsilon_x} \right]_{0}^{1},
\label{eq:7}
\end{equation}
where $[u]_{0}^{1}\triangleq\min\{1,\max\{0,u\}\}$, $\epsilon_x>0$, and the training statistics remain fixed during validation and testing.
\end{rmk}

\section{Low-Complexity Direct-Channel Tracking}

During the direct-channel pilot stage, the IRS is operated in its minimum-reflection state, and the active users in $\mathcal D_t$ transmit the pilot sequence $\boldsymbol{\phi}_1$. The O-CA algorithm is employed to recursively track the direct user-to-BS channels, whose predictions serve as prior information for the subsequent cascaded-channel acquisition and resource optimization procedures. According to \cite{SF23}, O-CA computes an $M\times 1$ channel state and also $M\times M$ covariance and cross-covariance matrices for each user  to avoid working directly with a large $KM\times KM$ covariance matrix. The O-CA algorithm consists of the following two steps:

{\bf{Prediction Step.}} For every pair of users  $k, l = 1, \dots, K$ and $k\neq l$, the prediction equations are as follows:
\begin{align}
&\widehat{\mathbf{h}}_{t|t-1}^{d, (k)}=\mathbf{A}_t^{d, (k)}\widehat{\mathbf{h}}_{t-1|t-1}^{d, (k)},\label{eq:8}\\ 
&\mathbf{P}_{t|t-1}^{d, (k,k)}=\mathbf{A}_t^{d, (k)}\mathbf{P}_{t-1|t-1}^{d, (k,k)}(\mathbf{A}_t^{d, (k)})^H+\mathbf{Q}_t^{d, (k)}, \label{eq:9}\\ 
&\mathbf{P}_{t|t-1}^{d, (k,l)}=\mathbf{A}_t^{d, (k)}\mathbf{P}_{t-1|t-1}^{d, (k,l)}(\mathbf{A}_t^{d, (l)})^H.\label{eq:10}
\end{align}
Specifically, the cross-covariance matrices capture the statistical coupling induced by pilot collisions among users sharing the same pilot sequence.

{\bf{Correction Step.}} For each active user $k\in\mathcal{D}_t$, the channel state and covariance matrices are updated as:
\begin{align}
  {\bf{\widehat{h}}}_{t|t}^{d, (k)}&={\bf{\widehat{h}}}_{t|t-1}^{d, (k)}+{\bf{K}}_{t}^{(k)} ({\bf{y}}_{t}- \sum_{j \in D_{t}}{\bf{\widehat{h}}}_{t|t-1}^{d, (j)}),  \label{eq:11}\\
  {\bf{P}}_{t|t}^{d, (k, k)}&={\bf{P}}_{t|t-1}^{d, (k, k)} + {\bf{K}}_{t}^{(k)}{\boldsymbol{\Upsilon}}_t({\bf{K}}_{t}^{(k)})^{H} +{\bf{K}}_{t}^{(k)}{\bf{R}}_{t} ({\bf{K}}_{t}^{(k)})^{H} \nonumber \\
&-\sum_{j \in D_{t}}{\bf{P}}_{t|t-1}^{d, (k, j)}({\bf{K}}_{t}^{(k)})^{H}-{\bf{K}}_{t}^{(k)}\sum_{n \in D_{t}}{\bf{P}}_{t|t-1}^{d, (n, k)}, \label{eq:12} \\
{\bf{P}}_{t|t}^{d, (k, l)}&={\bf{P}}_{t|t-1}^{d, (k, l)}-\sum_{m \in D_{t}}{\bf{P}}_{t|t-1}^{d, (k, m)}({\bf{K}}_{t}^{(l)})^{H}\nonumber \\
&-{\bf{K}}_{t}^{(k)}\sum_{n \in D_{t}}{\bf{P}}_{t|t-1}^{d, (n, l)} +{\bf{K}}_{t}^{(k)}{\boldsymbol{\Upsilon}}_t({\bf{K}}_{t}^{(l)})^{H}, \label{eq:13}
   \end{align}
where ${\boldsymbol{\Upsilon}}_t\triangleq\sum_{i\in\mathcal D_t}\sum_{j\in\mathcal D_t}\mathbf P_{t|t-1}^{d,(i,j)}$, and  O-CA gain for an active user $k\in\mathcal{D}_t$ is $\mathbf{K}_t^{(k)}=\sum_{j\in\mathcal{D}_t}\mathbf{P}_{t|t-1}^{d, (k,j)}
\big({\boldsymbol{\Upsilon}}_t+\mathbf{R}_t\big)^{-1}$.

\begin{rmk}
 Under time-invariant state-space and observation-noise statistics, the O-CA covariance matrices and  the estimator gain are bounded, converge to steady-state matrices \cite{SF23}, and can therefore be precomputed offline.
\end{rmk}
The O-CA tracker estimates only the direct channels. The resulting estimates are stacked in
$\widehat{\mathbf H}^{d}_{t|t}$ and used as the multiuser direct-channel anchor in the differential semi-blind acquisition of Section~IV-A.

\section{Semi-Blind Cascaded-Channel Acquisition and Beam-Domain Tracking}

The O-CA tracker provides the direct-channel estimates required for semi-blind cascaded-channel acquisition. Each slot contains $S$ symbols and is modeled as one coherence block. When recalibration is triggered at slot $t$, following the $\tau_p$-symbol direct-channel pilot, all active users repeat the same unknown $\alpha$-symbol payload under one minimum-reflection baseline and $N$ DFT-coded IRS states. The remaining regular data interval has length $\beta=S-\tau_p-\alpha(N+1)>0$. Baseline subtraction removes the direct-channel contribution before cascaded-channel recovery, after which the element-wise posterior is projected onto the reduced beam-domain states of Section~IV-C. When recalibration is not triggered, the existing cascaded states are propagated by prediction only.

\subsection{Differential Repeated-Data Acquisition}

Let $ \mathbf X_t = \operatorname{row}_{k\in\mathcal D_t} \big\{ \mathbf x_t^{(k)}
\big\} \in \mathbb C^{|\mathcal D_t|\times\alpha}$ denote the unknown payload block transmitted by the active users.  Its rows are zero mean, mutually independent across users, and satisfy $\mathbb E[\mathbf X_t\mathbf X_t^{H}] =\alpha\sigma_x^2\mathbf I$.
The same realization of $\mathbf X_t$ is repeated over all acquisition configurations. During the baseline block, the IRS is placed in its minimum-reflection state, modeled as zero reflection in the acquisition implementation. With $\mathbf H_{t}^{d} = \big[ \mathbf h_{t}^{d,(k)} \big]_{k\in\mathcal D_t}$, the received baseline block is $\mathbf Y_{t,0}=\mathbf H_{t}^{d}\mathbf X_t+\mathbf N_{t,0}$. Let $\mathbf F_N = [\mathbf f_1,\ldots,\mathbf f_N] \in\mathbb C^{N\times N}$ be the unit-modulus DFT pattern matrix with $[\mathbf F_N]_{n,p} =\exp\!\big(-j\frac{2\pi(n-1)(p-1)}{N}\big)$ and $\mathbf F_N\mathbf F_N^{H}=N\mathbf I_N$. For coded state $p=1,\ldots,N$, define $\mathbf C_{t,p} =\big[\mathbf L_t^{(k)}\mathbf f_p\big]_{k\in\mathcal D_t}$. The corresponding received block is
\begin{equation}
\mathbf Y_{t,p} = \big( \mathbf H_{t}^{d}+\mathbf C_{t,p} \big)\mathbf X_t + \mathbf N_{t,p}, \quad p=1,\ldots,N.   \label{eq:14}
\end{equation}

The direct and cascaded channels are assumed constant over the pilot, baseline, and all $N$ coded acquisition blocks within slot $t$. Because the same unknown payload realization is reused in all configurations, subtracting the baseline removes the direct channel  at the observation level:
\begin{equation}
\Delta\mathbf Y_{t,p} \triangleq \mathbf Y_{t,p}-\mathbf Y_{t,0} = \mathbf C_{t,p}\mathbf X_t + \Delta\mathbf N_{t,p}, \label{eq:15}
\end{equation}
where $\Delta\mathbf N_{t,p}=\mathbf N_{t,p}-\mathbf N_{t,0}$. Thus, the direct-channel term is removed before cascaded-channel recovery, rather than by subtracting two separately estimated channels. Let $\widehat{\mathbf H}_{t|t}^{d} = \big[ \widehat{\mathbf h}_{t|t}^{d,(k)} \big]_{k\in\mathcal D_t}$, $\widehat{\mathbf{R}}_{t} \triangleq (\widehat{\mathbf H}_{t|t}^{d})^{H} \widehat{\mathbf H}_{t|t}^{d}$, and $\nu\triangleq\sigma_n^2/\sigma_x^2$. We assume $|\mathcal D_t|\le M$, and 
$\widehat{\mathbf H}_{t|t}^{d}$ has full column rank, so that $\widehat{\mathbf{R}}_{t}$ is nonsingular. The receiver-noise matrices are independent across acquisition
configurations, with independent columns distributed as $\mathcal{CN}(\mathbf 0,\sigma_n^2\mathbf I_M)$. A normalized empirical user-data covariance matrix is estimated
using the baseline observation:
\begin{align}
&\mathbf Z_{t,0} = (\widehat{\mathbf{R}}_{t})^{-1} (\widehat{\mathbf H}_{t|t}^{d})^{H} \mathbf Y_{t,0}, \nonumber\\
&\widehat{\mathbf S}_{X,t} = \frac{ \mathbf Z_{t,0}\mathbf Z_{t,0}^{H} }{
\alpha\sigma_x^2 } - \nu \widehat{\mathbf{R}}_{t}^{-1}. \label{eq:16}
\end{align}

Before inversion, $\widehat{\mathbf S}_{X,t}$ is Hermitian symmetrized and projected onto the positive-definite cone using a small eigenvalue floor for numerical stability. For coded state $p$, the cascaded response is estimated as
\begin{align}
\widehat{\mathbf C}_{t,p} = \Big( \frac{ \Delta\mathbf Y_{t,p} \mathbf Y_{t,0}^H \widehat{\mathbf H}_{t|t}^{d} }{ \alpha\sigma_x^2 } + \nu\widehat{\mathbf H}_{t|t}^{d} \Big) \widehat{\mathbf R}_{t}^{-1} \widehat{\mathbf S}_{X,t}^{-1}.\label{eq:17}
\end{align}
 For user $k$, collect $\widehat{\mathbf C}_{t}^{(k)} =\big[\widehat{\mathbf c}_{t,1}^{(k)},\ldots, \widehat{\mathbf c}_{t,N}^{(k)} \big]$.  The $\nu\widehat{\mathbf H}_{t|t}^{d}$ term compensates for the baseline-noise component contained in $\Delta\mathbf N_{t,p}=\mathbf N_{t,p}-\mathbf N_{t,0}$. Since $\mathbf C_t^{(k)}=\mathbf L_t^{(k)}\mathbf F_N$, the element-wise semi-blind estimate follows from the inverse DFT $\widehat{\mathbf L}_{t}^{\mathrm{SB},(k)}
=\frac{1}{N}\widehat{\mathbf C}_{t}^{(k)}\mathbf F_N^{H}$. 

\begin{rmk}
Although \eqref{eq:15} removes the instantaneous direct component, the O-CA estimate is still required to separate the repeated multiuser data through the baseline observation. Define the direct-anchor mixing matrix as $\mathbf A_t^{\mathrm{anc}} \triangleq \Big(
(\widehat{\mathbf H}_{t|t}^{d})^{H} \widehat{\mathbf H}_{t|t}^{d} \Big)^{-1} (\widehat{\mathbf H}_{t|t}^{d})^{H} \mathbf H_t^d$.  Assuming $\mathbf A_t^{\mathrm{anc}}$ is nonsingular, the large-sample, noise-free limit of the $p$th DFT-state estimate is $\widehat{\mathbf C}_{t,p} \rightarrow \mathbf C_{t,p} \big(\mathbf A_t^{\mathrm{anc}}\big)^{-1}$. Thus, direct-anchor mismatch produces a nonvanishing cross-user mixing floor even when finite-sample and receiver-noise errors vanish. For a perfect anchor, $\widehat{\mathbf H}_{t|t}^{d}=\mathbf H_t^d$, so that $\mathbf A_t^{\mathrm{anc}}=\mathbf I$ and the floor disappears.
\end{rmk}

\subsection{Semi-Blind Posterior Initialization}

For a recalibration slot $t_r$, define the semi-blind estimation error as $\widetilde{\mathbf L}_{t_r}^{(k)} \triangleq \mathbf L_{t_r}^{(k)} - \widehat{\mathbf L}_{t_r}^{\mathrm{SB},(k)}$. Its covariance is $\mathbf C_{\mathrm{SB}}^{(k)} \triangleq \mathbb E \big[ \operatorname{vec} ( \widetilde{\mathbf L}_{t_r}^{(k)}) \operatorname{vec}( \widetilde{\mathbf L}_{t_r}^{(k)} )^H \big] \in\mathbb C^{MN\times MN}$, where the expectation is taken  over the channel realizations, transmitted data, and receiver noise. The covariance $\mathbf C_{\mathrm{SB}} ^{(k)}$ is estimated from repeated acquisition trials and captures finite-block recovery error, receiver noise, direct-channel-anchor error, and the effect of the regularized data-covariance estimate. It is computed offline and kept fixed during online control.

After recalibration, the full-state posterior is initialized as $\widehat{\mathbf L}_{t_r|t_r}^{(k)} = \widehat{\mathbf L}_{t_r}^{\mathrm{SB},(k)}$, $\widehat{\mathbf s}_{t_r|t_r}^{(k)} = \operatorname{vec} \big( \widehat{\mathbf L}_{t_r}^{\mathrm{SB},(k)} \big)$, and  $\overline{\mathbf P}_{t_r|t_r}^{(k,k)} = \mathbf C_{\mathrm{SB}}^{(k)}$.
The effect of cross-user mixing on each user's marginal acquisition error is captured empirically by $\mathbf C_{\mathrm{SB}}^{(k)}$. Cross-user error cross-covariances are not propagated in the subsequent prediction stage.

Let $\{\mathcal N_u^{e}\}_{u=1}^{U}$ be a fixed ordered partition of the IRS into contiguous elementary units, independent of the dynamic control partition, such that
$\mathcal N_u^{e}\cap\mathcal N_v^{e}=\emptyset$ for $u\neq v$, and $\bigcup_{u=1}^{U}\mathcal N_u^{e}=\mathcal N$. Let $\mathcal V=\{\mathbf v_q\}_{q=1}^{Q}$ denote the finite IRS control codebook, whose phase quantization is specified in Section~V. Define $\widetilde{\mathbf v}_{u,q} \triangleq \mathbf 1_{\mathcal N_u^e}\odot\mathbf v_q$, which retains the entries of $\mathbf v_q$ over $\mathcal N_u^e$ and is zero elsewhere. The posterior beam-domain state is then ${\mathbf r}_{u,t_r|t_r}^{(k)}(q) = \widehat{\mathbf L}_{t_r}^{\mathrm{SB},(k)} \widetilde{\mathbf v}_{u,q} = \widehat{\mathbf L}_{t_r}^{\mathrm{SB},(k)} (:,\mathcal N_u^{e}) [\mathbf v_q]_{\mathcal N_u^{e}}$. With $\mathbf H_{u,q} \triangleq \widetilde{\mathbf v}_{u,q}^{T}\otimes\mathbf I_M$, the corresponding posterior error covariance is
\begin{align}
\mathbf P_{r,u,q, t_r|t_r}^{ (k,k)}& = \mathbf H_{u,q} \mathbf C_{\mathrm{SB}}^{(k)} \mathbf H_{u,q}^{H}\nonumber \\ 
&= \left( \widetilde{\mathbf v}_{u,q}^{T}\otimes\mathbf I_M \right) \mathbf C_{\mathrm{SB}}^{(k)} \left( \widetilde{\mathbf v}_{u,q}^{*}\otimes\mathbf I_M \right). \label{eq:18}
\end{align}
These projected covariances can be precomputed for the tracked unit-beam pairs, avoiding online manipulation of the full $MN\times MN$ covariance matrix.

For reference, define the full cascaded state as $\mathbf s_t^{(k)} \triangleq \operatorname{vec}(\mathbf L_t^{(k)}) \in \mathbb{C}^{MN \times 1}$ with dynamics $\mathbf s_{t+1}^{(k)} = \overline{\mathbf A}_t^{(k)} \mathbf s_t^{(k)} + \overline{\mathbf u}_t^{(k)}$, where $\overline{\mathbf u}_t^{(k)} \sim \mathcal{CN} ( \mathbf 0, \overline{\mathbf Q}_t^{(k)})$. The matrices $\mathbf A_t^{(k)}$ and $\mathbf Q_t^{(k)}\in\mathbb C^{MN\times MN}$ are identified offline from long-term cascaded-channel observations. Between recalibration events, the full-state predictor follows
\begin{equation}
\begin{aligned}
\widehat{\mathbf s}_{t+1|t}^{(k)}&= \overline{\mathbf A}_t^{(k)} \widehat{\mathbf s}_{t|t}^{(k)},\\
\overline{\mathbf P}_{t+1|t}^{(k,k)} &= \overline{\mathbf A}_t^{(k)} \overline{\mathbf P}_{t|t}^{(k,k)} \big(\overline{\mathbf A}_t^{(k)}\big)^H + \overline{\mathbf Q}_t^{(k)}.
\end{aligned} \label{eq:19}
\end{equation}

The predicted cascaded-channel matrix is $\widehat{\mathbf L}_{t+1|t}^{(k)} = \operatorname{unvec}_{M\times N} \big( \widehat{\mathbf s}_{t+1|t}^{(k)}
\big)$. This full-state recursion serves as an analytical reference; the reduced-order implementation used for large $N$ is described next. The projected covariance in \eqref{eq:18} provides the raw post-acquisition uncertainty model for each tracked unit-beam state. Section~VII-A introduces its slot-dependent calibrated counterpart, which is used to estimate the residual uncertainty associated with a candidate semi-blind recalibration.

\subsection{Low-Complexity Beam-Domain Tracking}

Propagating the full $MN$-dimensional cascaded state becomes costly for large $N$. We thus track reflected responses over the fixed elementary units introduced in Section IV-B. Let $\mathcal Q_t^{\mathrm{tr}}\subseteq\{1,\ldots,Q\}$ denote the set of beam indices whose reduced-order responses are tracked at slot $t$, and define $\mathbf r_{u,t}^{(k)}(q) = \mathbf L_t^{(k)} \widetilde{\mathbf v}_{u,q} =(\widetilde{\mathbf v}_{u,q}^{T}\otimes\mathbf I_M) \mathbf s_t^{(k)}\in\mathbb C^{M \times 1}$. Since projection of the full-state process does not generally produce a closed reduced-order model, the unit--beam dynamics are identified directly from historical responses:
\begin{equation}
\mathbf r_{u,t+1}^{(k)}(q)=\mathbf A_{u,q,t}^{(k)}\mathbf r_{u,t}^{(k)}(q)+
\mathbf e_{u,q,t}^{(k)},\label{eq:20}
\end{equation}
where $\mathbf e_{u,q,t}^{(k)} \sim \mathcal{CN} \!\big( \mathbf0, \mathbf Q_{r,u,q,t}^{(k)} \big)$. The transition matrix $\mathbf A_{u,q,t}^{(k)}$ and innovation covariance $\mathbf Q_{r,u,q,t}^{(k)}$ are estimated offline from historical unit-level beam responses and their one-step prediction residuals. The resulting predictor is
\begin{equation}
\begin{aligned}
\widehat{\mathbf r}_{u,t+1|t}^{(k)}(q)&= \mathbf A_{u,q,t}^{(k)} \widehat{\mathbf r}_{u,t|t}^{(k)}(q),\\ 
\mathbf P_{r,u,q,t+1|t}^{(k,k)} &= \mathbf A_{u,q,t}^{(k)} \mathbf P_{r,u,q,t|t}^{(k,k)}
\big(\mathbf A_{u,q,t}^{(k)}\big)^H + \mathbf Q_{r,u,q,t}^{(k)}. \end{aligned} \label{eq:21}
\end{equation}

For structured online implementation, the covariance sequences are precomputed over the admissible prediction ages, so that online tracking propagates only the \(M\)-dimensional mean vector and retrieves the corresponding prediction-error covariance from a lookup table. If the same beam set $\mathcal Q_t^{\mathrm{tr}}$ is tracked for all $U$ elementary units, the reduced state dimension per user is $MU|\mathcal Q_t^{\mathrm{tr}}|$, compared with $MN$ for the full cascaded state. Hence, dimensionality is reduced when $U|\mathcal Q_t^{\mathrm{tr}}|\ll N$. For a control partition $g$ composed of the elementary-unit index set $\mathcal U_{g,t}$, the predicted reflected response is $\widehat{\mathbf r}_{g,t+1|t}^{(k)}(q) =\sum_{u\in\mathcal U_{g,t}} \widehat{\mathbf r}_{u,t+1|t}^{(k)}(q)$. To avoid propagating cross-unit error covariances, prediction errors of distinct elementary units are assumed approximately uncorrelated. Accordingly, $\mathbf P_{r,g,q,t+1|t}^{(k,k)} \approx \sum_{u\in\mathcal U_{g,t}} \mathbf P_{r,u,q,t+1|t}^{(k,k)}$.

The same reduced-order state representation is used after recalibration and during prediction-only operation. Under recalibration, the semi-blind posterior of Section~IV-B is projected onto the tracked unit--beam states and used as the updated beam-domain posterior; otherwise, the existing posterior is propagated according to \eqref{eq:21}. The mode selection remains part of the joint control problem in Section~VI and is resolved by the structured recalibration rule in Section~VII-A.

\section{Beam-Dependent Predicted Effective Channel}

To account for practical finite-resolution hardware, the IRS phase configuration is selected as $\mathbf v_q=[e^{j\psi_{1,q}},\ldots,e^{j\psi_{N,q}}]^T$ for $q=1,\ldots,Q$. For $b$-bit phase quantization, the phase of element $n$ is restricted to $\psi_{n,q}\in \big\{ \frac{2\pi m}{2^b}:m=0,\ldots,2^b-1 \big\}$, where at most $2^{bN}$ full-surface configurations are possible, and the codebook may contain any subset of $Q\le 2^{bN}$ configurations.  These unit-modulus codewords are used during normal reflection and are distinct from the minimum- reflection acquisition state of Section~IV-A.

Let \(\mathcal N_{g,t}\subseteq\mathcal N\) denote the element set of IRS unit  \(g\) with \(n_{g,t}\triangleq|\mathcal N_{g,t}|\). The ordered sets $\{\mathcal N_{g,t}\}_{g=1} ^{G_t}$ form the current IRS partition. For unit \(g\), let \(\mathbf v_{g,q}\triangleq [\mathbf v_q]_{\mathcal N_{g,t}}\) denote the restriction of codeword \(\mathbf v_q\) to that unit. The executable controller uses the equal-size candidate partitions defined in Section~VII-B. Under the intended-user dominant-reflection model, only the IRS units assigned to user \(k\) are retained in the controller-side effective channel, while the reflected leakage from units assigned to other users is neglected. This approximation is most accurate when unit-level beam focusing suppresses the unintended reflected components relative to the desired component. For the beam-index vector \(\mathbf q_t=[q_{1,t},\ldots,q_{G_t,t}]\) and assignment matrix \(\mathbf B_t=[a_{g,t}^{(k)}]\in\{0,1\}^{G_t\times K}\), the one-step predicted effective channel is 
\begin{equation}
 \widehat{\mathbf h}_{t+1|t}^{\mathrm{eff},(k)} (\mathbf B_t,\mathbf q_t) \approx \widehat{\mathbf h}_{t+1|t}^{d,(k)} + \sum_{g=1}^{G_t} a_{g,t}^{(k)} \widehat{\mathbf r}_{g,t+1|t}^{(k)}(q_{g,t}), \label{eq:22}
 \end{equation}
 where $\widehat{\mathbf r}_{g,t+1|t}^{(k)}(q) = \sum_{u\in\mathcal U_{g,t}} \widehat{\mathbf r}_{u,t+1|t}^{(k)}(q)$ is obtained from the reduced-order predictor in Section IV-C.

In the beam-domain, the predicted effective channel is directly formed from the aggregated unit-level responses, and its covariance is obtained from the corresponding reduced-order prediction covariances.  The user-specific masked codeword $\widetilde{\mathbf v}_t^{(k)} \triangleq \operatorname{col}_{g=1}^{G_t} \big( a_{g,t}^{(k)} \mathbf v_{g,q_{g,t}} \big)\in\mathbb C^{N}$ is the full-length codeword whose entries corresponding to IRS units not assigned to user $k$ are set to zero.  Applying the identity \(\mathbf L\mathbf v=(\mathbf v^T\otimes\mathbf I_M)\operatorname{vec}(\mathbf L)\), the predicted effective channel error is
\begin{equation}
\Delta\mathbf h_{t+1|t}^{\mathrm{eff},(k)} =\Delta\mathbf h_{t+1|t}^{d,(k)}
+ \big(\widetilde{\mathbf v}_t^{(k)T} \otimes\mathbf I_M \big)
\Delta\mathbf s_{t+1|t}^{(k)}.\label{eq:23}
\end{equation}

Define $\Delta\mathbf h_{t+1|t}^{d,(k)} \triangleq \mathbf h_{t+1}^{d,(k)} - \widehat{\mathbf h}_{t+1|t}^{d,(k)}$ and $\Delta\mathbf s_{t+1|t}^{(k)} \triangleq \operatorname{vec} \big( \mathbf L_{t+1}^{(k)} - \widehat{\mathbf L}_{t+1|t}^{(k)} \big)$. Conditioned on the available information at slot $t$, and assuming that the direct and cascaded-channel prediction errors are zero mean and mutually uncorrelated, the effective-channel covariance is then
\begin{equation}
\mathbf P_{t+1|t}^{\mathrm{eff},(k,k)} (\mathbf B_t,\mathbf q_t) = \mathbf P_{t+1| t}^{d,(k,k)} + \sum_{g=1}^{G_t} a_{g,t}^{(k)} \mathbf P_{r,g,q_{g,t},t+1|t}^{(k,k)}, \label{eq:24}
\end{equation}
The effective-channel uncertainty, defined as $ \eta_t^{(k)}(\mathbf B_t,\mathbf q_t) \triangleq \operatorname{Tr} \big(\mathbf P_{t+1|t}^{\mathrm{eff},(k,k)} (\mathbf B_t,\mathbf q_t) \big)$, quantifies the total mean-square predicted effective channel error across the BS antennas.

\section{Problem Formulation}

The slot-$t$ control problem jointly selects the recalibration mode and the remaining IRS and radio-resource variables. For a candidate recalibration decision $a_t^{\mathrm{cal}}\in\{0,1\}$, let $\widehat{\mathbf h}_{t+1|t}^{\mathrm{eff},(k),\mathrm{ctrl}} (\mathbf B_t,\mathbf q_t;a_t^{\mathrm{cal}})$ and $\mathbf P_{t+1|t}^{\mathrm{eff},(k,k),\mathrm{ctrl}} (\mathbf B_t,\mathbf q_t;a_t^{\mathrm{cal}})$ denote the effective-channel mean and covariance obtained from the posterior associated with the candidate recalibration mode.

{\emph{Predicted SINR and Net Rate:}} Let $p^{(k)}$ denote the transmit power of user $k$, and let $\mathbf w_{t+1|t}^{(k)}\in\mathbb C^M$ denote the receive combiner constructed from the control-conditioned predicted effective channels, and let $\mathbf c_t
\triangleq \big[c_t^{(k)}\big]_{k\in\mathcal D_t}$. The simulations use zero-forcing (ZF) combining. For a candidate action $(\mathbf B_t,\mathbf q_t,\mathbf c_t, a_t^{\mathrm{cal}})$, the uncertainty-aware predicted SINR in \eqref{eq:25} accounts for channel-prediction errors as additional uncorrelated interference.
\begin{figure*}[t!]
\normalsize
\begin{equation}
\widehat{\mathrm{SINR}}_{t+1|t}^{(k)}= \frac{ p^{(k)} \big| (\mathbf w_{t+1| t}^{(k)})^{H} \widehat{\mathbf h}_{t+1|t}^{\mathrm{eff},(k),\mathrm{ctrl}}\big|^2
}{\begin{aligned}  \sigma_n^2 \big\|\mathbf w_{t+1|t}^{(k)}\big\|_2^2&+p^{(k)}
(\mathbf w_{t+1|t}^{(k)})^{H}\mathbf P_{t+1|t}^{\mathrm{eff}, (k,k),\mathrm{ctrl}}
\mathbf w_{t+1|t}^{(k)}\\
&+\displaystyle \sum_{\substack{j\in\mathcal D_t\setminus\{k\}\\  c_t^{(j)}=c_t^{(k)} }} p^{(j)}\Big[\big|(\mathbf w_{t+1|t}^{(k)})^{H}\widehat{\mathbf h}_{t+1|t}^{\mathrm{eff},(j),\mathrm{ctrl}}\big|^2+(\mathbf w_{t+1|t}^{(k)})^{H}\mathbf P_{t+1|t}^{\mathrm{eff},(j,j),\mathrm{ctrl}} \mathbf w_{t+1|t}^{(k)} \Big]\end{aligned}} \label{eq:25}
\end{equation}
\hrulefill
\end{figure*}

The action arguments of $\widehat{\mathrm{SINR}}_{t+1|t}^{(k)}$ are suppressed hereafter for compactness. The corresponding predicted rate over the regular data interval is
\begin{equation}
\widehat R_{t+1|t}^{\mathrm N,(k)} = B_{\mathrm{ch}} \log_2 \big( 1+\widehat{\mathrm{SINR}}_{t+1|t}^{(k)} \big), \label{eq:26} 
\end{equation}
where $B_{\mathrm{ch}}$ is the bandwidth of one uplink channel. As the same unknown $\alpha$-symbol payload block is reused over the baseline and the $N$ DFT-coded acquisition states, the coded repetitions carry no additional fresh payload. Hence, the net block rate is
\begin{equation}
\begin{aligned}
\widehat R_t^{\mathrm{net},(k)} ={}& (1-a_t^{\mathrm{cal}}) \frac{S-\tau_p}{S}
\widehat R_{t+1|t}^{N,(k)} \\
&+ a_t^{\mathrm{cal}} \frac{ \alpha \widehat R_{t+1| t}^{\mathrm{base},(k)} +\beta \widehat R_{t+1|t}^{N,(k)} }{S},
\end{aligned}\label{eq:net_rate}
\end{equation}
where $\beta=S-\tau_p-\alpha(N+1)>0$, and $\widehat R_{t+1|t}^{\mathrm{base},(k)}$ denotes the predicted rate under the minimum-reflection baseline. Thus, $\widehat R_t^{\mathrm{net},(k)}$ is a one-step-ahead prediction-based utility that directly accounts for the acquisition overhead incurred by $a_t^{\mathrm{cal}}$.

{\emph{Fairness and Outage Metrics:}} Define the predicted instantaneous delivered rate as  $\widetilde R_t^{(k)} \triangleq \mathbf 1_{\{k\in\mathcal D_t\}} \widehat R_t^{\mathrm{net},(k)}$. The candidate-dependent smoothed rate is $\widehat{\overline R}_{t+1|t}^{(k)} = (1-\alpha_R)\overline R_t^{(k)} + \alpha_R\widetilde R_t^{(k)}$, where $0<\alpha_R<1$. The regularized Jain fairness index is $\widehat J_{t+1|t}^{\mathrm{fair}} \triangleq \frac{ \big( \sum_{k=1}^{K} \widehat{\overline R}_{t+1|t}^{(k)} \big)^2 }{ K\sum_{k=1}^{K} \big( \widehat{\overline R}_{t+1|t}^{(k)} \big)^2 +\epsilon_J }$, where $\epsilon_J>0$. Inactive users contribute zero instantaneous rate. For a target rate $R_{\min}>0$, the predicted outage count is $\widehat N_t^{\mathrm{out}} \triangleq \sum_{k\in\mathcal D_t} \mathbf 1 \big\{ \widehat R_t^{\mathrm{net},(k)}<R_{\min} \big\}$. The dependence of these metrics on the candidate action, including $a_t^{\mathrm{cal}}$, is omitted for notational simplicity.

{\emph{Joint Control Problem:}} We formulate the equal-size partition family used by the executable controller. Let $\mathcal S_N \subseteq \left\{ n\in\mathbb Z_{>0}: n\mid N,\; \Delta_N\mid n \right\}$ contain the divisors of $N$ that are compatible with the elementary tracking grid of Section~IV-B. For $N_{0,t}\in\mathcal S_N$, define $G_t=N/N_{0,t}$ contiguous equal-size units, $\mathcal N_{g,t} = \{ (g-1)N_{0,t}+1,\ldots,gN_{0,t} \}, \qquad g=1,\ldots,G_t$. The full IRS reflection vector and phase-switching cost are $\boldsymbol\nu_t = \operatorname{col}_{g=1}^{G_t} ( \mathbf v_{g,q_{g,t}})$ and  $C_t^{\mathrm{sw}} = \sum_{n=1}^{N} \mathbf 1 \left\{ [\boldsymbol\nu_t]_n \neq
[\boldsymbol\nu_{t-1}]_n \right\}$, respectively. The slot-$t$ control problem is
\begin{subequations}
\begin{align}
\max_{\substack{N_{0,t}\in\mathcal S_N,\\\mathbf q_t,\mathbf c_t,\mathbf B_t,a_t^{\mathrm{cal}} }}& \sum_{k\in\mathcal D_t} \widehat R_{t}^{\mathrm{net}, (k)} +\lambda_F\widehat J_{t+1|t}^{\mathrm{fair}} -\lambda_O\widehat N_t^{\mathrm{out}} -\lambda_{\mathrm{sw}}C_t^{\mathrm{sw}}\nonumber\\ 
&-\lambda_U \sum_{k\in\mathcal D_t} \operatorname{Tr} \bigl( \mathbf P_{t+1| t}^{\mathrm{eff},(k,k),\mathrm{ctrl}} (\mathbf B_t,\mathbf q_t;a_t^{\mathrm{cal}}) \bigr) \label{eq:28a}\\
\mathrm{s.t.}\quad&q_{g,t}\in\mathcal Q_t^{\mathrm{tr}}, \quad g=1,\ldots,G_t,\label{eq:28b}\\
&c_t^{(k)}\in\mathcal C,\quad k\in\mathcal D_t,\label{eq:28c}\\
&a_t^{\mathrm{cal}}\in\{0,1\},\quad a_{g,t}^{(k)}\in\{0,1\},\label{eq:28d}\\
&a_{g,t}^{(k)}=0, \quad g=1,\ldots,G_t,\; k\notin\mathcal D_t ,\label{eq:28e}\\
&\sum_{k\in\mathcal D_t} a_{g,t}^{(k)}=1,\quad g=1,\ldots,G_t,\label{eq:28f}\\
&G_t=\frac{N}{N_{0,t}},\label{eq:28g}\\
&\sum_{\tau=\max\{1,t-W+1\}}^{t}N_{0,\tau}\sum_{g=1}^{G_\tau}a_{g,\tau}^{(k)}
\geq s_{\min}^{(k)},\quad k\in\mathcal D_t,\label{eq:28h}
\end{align}
\end{subequations}
where $a_{g,t}^{(k)}=1$ indicates that IRS unit $g$ is assigned to user $k$, so that $N_t^{(k)} = N_{0,t}  \sum_{g=1}^{G_t}a_{g,t}^{(k)}$ is the corresponding allocated aperture. Consistent with Section~VII-A, if an active user has no initialized cascaded-channel state, the feasible action set is restricted to $a_t^{\mathrm{cal}}=1$.

Constraint~\eqref{eq:28h} prevents persistent IRS-service starvation by enforcing a minimum cumulative element allocation over the latest $W$ slots. The thresholds $s_{\min}^{(k)}$ are chosen to preserve assignment feasibility, with the corresponding slot-wise condition given in Proposition~\ref{prop:irs_feasibility}. The service constraint regulates IRS access, whereas $\widehat J_{t+1|t}^{\mathrm{fair}}$ controls long-term rate fairness. The nonnegative coefficients  $\lambda_F$, $\lambda_O$, $\lambda_{\mathrm{sw}}$, and $\lambda_U$ define a reference scalarization of the coupled design objectives. They remain fixed during the testing. The resulting problem is temporally coupled, mixed-integer, and nonconvex, with a combinatorial joint action space. Thus, it is not solved directly by the executable controller.

\section{Uncertainty-Aware Structured Control}
\label{sec:structured_control}

The coupled variables in Section~VI are resolved through a deterministic structured procedure based on the channel statistics and uncertainty-aware physical metrics developed in Sections~III-VI. The controller first selects the recalibration mode and then, conditioned on this decision, determines the equal-size IRS granularity, unit ownership, and beam indices. Uplink-channel allocation is subsequently performed by the receiver-consistent procedure of Section~IX.

\subsection{Recalibration and IRS Feasibility}

To avoid solving the combinatorial resource problem in \eqref{eq:28a} for both recalibration modes, $a_t^{\mathrm{cal}}$ is determined before the remaining resource variables. If any active user has no initialized cascaded-channel state, semi-blind acquisition is mandatory and $a_t^{\mathrm{cal}}=1$. Otherwise, the following risk test is applied to the initialized active users. For each \(k\in\mathcal D_t\), let \(\widehat{\mathbf r}_{u,t}^{\mathrm{pre},(k)}(q)\) and \(\mathbf P_{r,u,q,t}^{\mathrm{pre},(k,k)}\) denote the beam-domain mean and covariance available immediately before the recalibration decision. Define
\begin{align}
U_{\mathrm{pre},t}^{(k)} \triangleq  (D_t^{(k)})^{-1}\displaystyle \sum_{u=1}^{U} \sum_{q\in\mathcal Q_t^{\mathrm{tr}}} \operatorname{Tr} \big( \mathbf P_{r,u,q,t}^{\mathrm{pre},(k,k)} \big) . \label{eq:pre_recal_risk_user}
\end{align}
where $D_t^{(k)} \triangleq \sum_{u=1}^{U} \sum_{q\in\mathcal Q_t^{\mathrm{tr}}} \big(\| \widehat{\mathbf r}_{u,t}^{\mathrm{pre},(k)}(q) \|_2^2 +\varepsilon_{\mathrm{cas}} \big)$ and $\varepsilon_{\mathrm{cas}}>0$ ensures numerical stability at low reflected-channel energy.  Let $\mathbf P_{0,u,q}^{(k,k)}$ denote the raw projected post-acquisition covariance in \eqref{eq:18}. Its calibrated slot-dependent counterpart is
\begin{equation}
\mathbf P_{\mathrm{floor},u,q,t}^{(k,k)} \triangleq \mathcal C_k\!\Big(10\log_{10} \big( \max \big\{ \operatorname{Tr} ( \mathbf P_{t|t}^{d,(k,k)} ), \varepsilon_d  \big\} \big)\Big)\mathbf P_{0,u,q}^{(k,k)}, \label{eq:conditional_floor_scale}
\end{equation}
where $\varepsilon_d>0$ and \(\mathcal C_k(\cdot)\) is obtained offline from user-specific calibration data. Online queries are clipped to the calibrated range, interpolated in dB, and converted to a linear covariance scale. The predicted post-acquisition covariance-floor risk is
\begin{equation}
U_{\mathrm{floor},t}^{(k)}\triangleq (D_t^{(k)})^{-1}\displaystyle \sum_{u=1}^{U} \sum_{q\in\mathcal Q_t^{\mathrm{tr}}} \operatorname{Tr} \big( \mathbf P_{\mathrm{floor},u,q,t}^{(k,k)} \big). \label{eq:post_recal_risk_user}
\end{equation}

Define the worst-user system risks as $U_{\mathrm{pre},t} \triangleq \max_{k\in\mathcal D_t} U_{\mathrm{pre},t}^{(k)}$ and $U_{\mathrm{floor},t} \triangleq \max_{k\in\mathcal D_t} U_{\mathrm{floor},t}^{(k)}$. For the dB  reliability threshold $\varepsilon_{\mathrm{cal}}^{\mathrm{dB}}$, let $\varepsilon_{\mathrm{cal}} \triangleq 10^{\varepsilon_{\mathrm{cal}} ^{\mathrm{dB}}/10}$.  The recalibration decision is
\begin{equation}
a_t^{\mathrm{cal}}=\mathbf 1 \left\{ U_{\mathrm{pre},t}>\varepsilon_{\mathrm{cal}}
\;\land\; U_{\mathrm{pre},t}>U_{\mathrm{floor},t} \right\}, \label{eq:actionable_recalibration}
\end{equation}
The reliability test $U_{\mathrm{pre},t}>\varepsilon_{\mathrm{cal}}$ identifies an unreliable prediction, whereas $U_{\mathrm{pre},t}>U_{\mathrm{floor},t}$ requires a positive predicted uncertainty reduction. If $a_t^{\mathrm{cal}}=1$, the semi-blind acquisition of
Section~IV-A refreshes the slot-$t$ beam-domain posterior according to Section~IV-B; otherwise, the prediction-only posterior is retained. The selected posterior is used for the subsequent resource decisions and as the state reference for the next decision epoch.

Let \(\mathcal{S}_N\) denote the admissible set of equal IRS-unit sizes. For a candidate \(n_0\in\mathcal{S}_N\), the IRS contains $G(n_0)=\frac{N}{n_0}$ contiguous units,  $\mathcal{N}_g(n_0) =\left\{(g-1)n_0+1,\ldots,gn_0\right\},\qquad g=1,\ldots,G(n_0)$.
Since every admissible $n_0$ is compatible with the fixed elementary tracking grid of Section~IV-C, these control units are formed by aggregating elementary unit-beam states and require no new cascaded-channel estimation.  Let \(\mathcal H_t^{W-1}\) denote the indices of the most recent \(W-1\) nonempty controller blocks preceding block \(t\).  The residual aperture deficit of active user $k$ is
\begin{equation}
d_t^{(k)} \triangleq \big[ s_{\min}^{(k)} - \sum_{\tau\in\mathcal H_t^{W-1}} N_{\tau}^{(k)} \big]^+, \quad k\in\mathcal D_t. \label{eq:service_deficit}
\end{equation}
where $[x]^+\triangleq\max\{x,0\}$ and $N_{\tau}^{(k)}$ is the previously allocated IRS aperture defined in Section~VI.  For candidate size $n_0$, the minimum number of units required by user $k$ is $n_{\mathrm{req},t}^{(k)}(n_0) \triangleq \Big\lceil \frac{d_t^{(k)}}{n_0} \Big\rceil$. Hence, the feasible candidate-size set is 
\begin{equation}
\mathcal{S}_{N,t}^{\mathrm{feas}} =\Big\{ n_0\in\mathcal{S}_N: \sum_{k\in\mathcal{D}_t} n_{\mathrm{req},t}^{(k)}(n_0) \leq G(n_0) \Big\}.\label{eq:feasible_unit_sizes}
\end{equation}

For $\mathcal S_{N,t}^{\mathrm{feas}}\neq\varnothing$, \eqref{eq:feasible_unit_sizes} is necessary and sufficient for an exclusive equal-size assignment satisfying the current residual service requirements. Its sequential preservation is proved in Section~X-A. Infeasible epochs, if any, are handled explicitly in the numerical evaluation.   For a control unit formed from several elementary units, the predicted mean responses are summed, whereas its acquisition-time uncertainty uses the directly calibrated contiguous-block covariance of Section~IV-C rather than a sum of elementary post-acquisition covariances.

\subsection{IRS Configuration Selection}

For a fixed $n_0\in\mathcal S_{N,t}^{\mathrm{feas}}$, IRS ownership is determined from an  isolated system-rate contribution. Let $\mathcal Q_t^{\mathrm{cand}}\subseteq\mathcal Q_t^{\mathrm{tr}}$ denote the beam-codeword set considered by the controller.  Since uplink channels have not yet been assigned, each IRS candidate is evaluated with all active users treated as a common provisional cochannel set, and the IRS-ranking metric is
\begin{equation}
\widehat{\mathcal R}_t(\mathbf B,\mathbf q) \triangleq \sum_{k\in\mathcal D_t}
B_{\mathrm{ch}} \log_2 \Big( 1+ \widehat{\mathrm{SINR}}_{t+1|t}^{\mathrm{IRS},(k)} (\mathbf B,\mathbf q) \Big), \label{eq:irs_reference_sumrate}
\end{equation}
where $\widehat{\mathrm{SINR}}_{t+1|t}^{\mathrm{IRS},(k)}$ is obtained from \eqref{eq:25} using the corresponding control-conditioned ZF receiver.  This metric is used only for IRS ranking; the channel-dependent rates are recomputed in Section~IX. Since $a_t^{\mathrm{cal}}$  is fixed before IRS optimization, the common acquisition-overhead factor does not affect this ranking.

Let $\widehat{\mathcal R}_{t,g\rightarrow k}^{\mathrm{iso}}(q;n_0)$ denote \eqref{eq:irs_reference_sumrate} evaluated when only unit $g$ of partition $n_0$ contributes to the reflected channel of user $k$ and uses beam $q$; all other reflected contributions are masked out. Define
\begin{align}
\Delta\widehat{\mathcal R}_{t,g}^{(k)}(n_0)&\triangleq \max_{q\in\mathcal Q_t^{\mathrm{cand}}} \widehat{\mathcal R}_{t,g\rightarrow k}^{\mathrm{iso}}(q;n_0)
- \widehat{\mathcal R}_t^{\mathrm{dir}}, \label{eq:isolated_unit_score} 
\\ q_{t,g}^{\mathrm{iso},(k)}(n_0)  &\in \arg\max_{q\in\mathcal Q_t^{\mathrm{cand}}}
\widehat{\mathcal R}_{t,g\rightarrow k}^{\mathrm{iso}}(q;n_0), \label{eq:isolated_beam}
\end{align}
where $\widehat{\mathcal R}_t^{\mathrm{dir}}$ denotes \eqref{eq:irs_reference_sumrate} with all reflected contributions removed. Thus, \(\Delta\widehat{\mathcal R}_{t,g}^{(k)}\) measures the predicted system-rate change produced by assigning the considered unit to user \(k\) under its best isolated beam.

The units are assigned sequentially. Let $n_{t,g-1}^{(j)}(n_0) = \sum_{\ell=1}^{g-1}
\mathbf 1 \big\{ k_{\ell,t}^{\star}(n_0)=j \big\}$ denote the number of units already assigned to user \(j\) before processing unit \(g\). If the current unit were assigned to candidate user \(k\), the residual demand of user $j$ would be
\begin{equation}
\delta_{t,g}^{(j)}(k;n_0)\triangleq\left[n_{\mathrm{req},t}^{(j)}(n_0) - n_{t,g-1}^{(j)}(n_0) - \mathbf 1\{j=k\} \right]^+ .\label{eq:residual_unit_demand}
\end{equation}
The residual-feasible owner set is therefore
\begin{equation}
\mathcal{K}_{t,g}^{\mathrm{feas}}(n_0)=\Big\{k\in\mathcal{D}_t:\sum_{j\in\mathcal{D}_t} \delta_{t,g}^{(j)}(k;n_0) \leq G(n_0)-g \Big\}. \label{eq:feasible_owner_set}
\end{equation}
The owner of unit \(g\) is selected as
\begin{equation}
k_{g,t}^{\star}(n_0) \in\arg\max_{k\in \mathcal{K}_{t,g}^{\mathrm{feas}}(n_0)} \Delta\widehat{\mathcal R}_{t,g}^{(k)}(n_0). \label{eq:deterministic_owner}
\end{equation}
Units are processed in the fixed order $g=1,\ldots,G(n_0)$, with ties in the isolated-beam and owner selections resolved by the smallest beam and user indices, respectively. Condition \eqref{eq:feasible_owner_set} preserves residual service feasibility, as established in Section~X-A.

The resulting owner indices define $[\widehat{\mathbf B}_t(n_0)]_{g,k} =\mathbf 1\{k=k_{g,t}^{\star}(n_0)\}$, and the corresponding beam initialization is
\begin{equation}
q_{g,t}^{(0)}(n_0)=q_{g,t}^{\mathrm{iso},\left(k_{g,t}^{\star}(n_0)\right)}(n_0),
\quad g=1,\ldots,G(n_0). \label{eq:isolated_beam_initialization}
\end{equation}
The beam vector is then refined by one fixed-order coordinate sweep:
\begin{equation}
q_{g,t}^{(1)}(n_0) \in \arg\max_{q\in\mathcal Q_t^{\mathrm{cand}}} \widehat{\mathcal R}_t \big( \widehat{\mathbf B}_t(n_0), \mathbf q_{t,g}^{(q)}(n_0) \big). \label{eq:coordinate_beam_refinement}
\end{equation}
where
\[
\begin{aligned}
\mathbf q_{t,g}^{(q)}(n_0) = \big[& q_{1,t}^{(1)}(n_0),\ldots, q_{g-1,t}^{(1)}(n_0), q, \\
& q_{g+1,t}^{(0)}(n_0),\ldots, q_{G(n_0),t}^{(0)}(n_0) \big]^T.
\end{aligned}
\] Ties are resolved by the smallest maximizing beam index. Since the incumbent beam belongs to \(\mathcal Q_t^{\mathrm{cand}}\), each coordinate update is nondecreasing in \(\widehat{\mathcal R}_t\). 

After the sweep, define $\widehat{\mathbf q}_t(n_0) =[q_{1,t}^{(1)}(n_0),\ldots,q_{G(n_0),t}^{(1)}(n_0)]^T$.  The final equal-size granularity is
\begin{equation}
N_{0,t}^{\star} = \min \arg\max_{n_0\in\mathcal S_{N,t}^{\mathrm{feas}}} \widehat{\mathcal R}_t \big( \widehat{\mathbf B}_t(n_0), \widehat{\mathbf q}_t(n_0) \big), \label{eq:final_unit_size_selection}
\end{equation} 
where the IRS candidates are ranked by the uncertainty-aware rate surrogate in \eqref{eq:irs_reference_sumrate}, with service feasibility enforced by \eqref{eq:feasible_unit_sizes} and \eqref{eq:feasible_owner_set}. The selected IRS variables are $\mathbf B_t^{\star} = \widehat{\mathbf B}_t(N_{0,t}^{\star})$ and $\mathbf q_t^{\star} = \widehat{\mathbf q}_t(N_{0,t}^{\star})$, yielding $\mathcal A_t^{\mathrm{IRS}} = (a_t^{\mathrm{cal}},N_{0,t}^{\star}, \mathbf B_t^{\star},\mathbf q_t^{\star})$.
The final uplink-channel vector $\mathbf c_t^{\star}$ is then obtained by the receiver-consistent sequential allocator in Section~IX.

\section{Full-Slot Deterministic Controller Execution}

The structured controller of Sections~VII-IX is executed once per nonempty physical control block. All decisions are causal and use only controller-available channel statistics, the IRS-service history, and parameters fixed during development. A recalibration block remains an active control block: when semi-blind acquisition is selected, the cascaded-channel posterior is first refreshed and the resulting posterior is then used to determine the IRS and uplink-channel actions for the remaining data interval. Thus, both normal and recalibration blocks advance the service and delivered-rate histories. Consistent with Section~IV, one physical control block corresponds to one control slot.

\subsection{Control-State and Resource Selection}

At the beginning of block $t$, the O-CA tracker updates the direct-channel posterior and the prediction-only beam-domain states. The controller then selects $a_t^{\mathrm{cal}}$ according to \eqref{eq:actionable_recalibration}, with recalibration being mandatory for any active user without an initialized cascaded state. If $a_t^{\mathrm{cal}}=1$, the semi-blind posterior is refreshed and projected onto the tracked unit-beam states before the remaining resource decisions are made.

Let \(\mathcal P_t^{\mathrm{ctrl}}\) denote the selected control-conditioned channel statistics. Subsequent resource decisions use the refreshed channel posterior when \(a_t^{\mathrm{cal}}=1\) and the prediction-only statistics otherwise. Let \(\mathcal P_t^{\mathrm{acq}}\) and \(\mathcal P_t^{\mathrm{pred}}\) denote the control-conditioned channel statistics obtained after the semi-blind refresh and prediction-only channel statistics, respectively. The selected collection provides the effective-channel means and error covariances used by the subsequent IRS and uplink-channel decisions.
Conditioned on $\mathcal P_t^{\mathrm{ctrl}}$, Section~VII-B selects $N_{0,t}^{\star}$, $\mathbf B_t^{\star}$, and $\mathbf q_t^{\star}$, while Section~IX determines $\mathbf c_t^{\star}$. The executed action is $\mathcal A_t=
\left( a_t^{\mathrm{cal}}, N_{0,t}^{\star},  \mathbf B_t^{\star}, \mathbf q_t^{\star}, \mathbf c_t^{\star} \right)$. 

\subsection{Block Execution and State Updates}

For \(a_t^{\mathrm{cal}}=0\), the selected IRS and uplink-channel configuration is applied over the regular data interval following the pilot. For \(a_t^{\mathrm{cal}}=1\), the minimum-reflection baseline and the \(N\) DFT-coded repeated-data states are first used for cascaded-channel acquisition, after which the selected action is applied over the remaining \(\beta\)-symbol regular-data interval.

The predicted net block rate is evaluated using \eqref{eq:net_rate}, so the acquisition overhead is charged to the same physical block in which recalibration is performed. In particular, only one realization of the repeated unknown \(\alpha\)-symbol payload is counted as fresh data; the DFT-coded repetitions provide channel-acquisition observations but no additional fresh payload. Thus, recalibration modifies both the channel-state accuracy and the usable payload duration, and no additional heuristic recalibration penalty is required. The corresponding numerical payload fractions are given in Section~XII.

Let $G_t^{\star}=\frac{N}{N_{0,t}^{\star}}$. The executed IRS aperture assigned to user \(k\) is $N_t^{(k)} = N_{0,t}^{\star} \sum_{g=1}^{G_t^{\star}} [\mathbf B_t^{\star}] _{g,k}$. After each nonempty block, \(N_t^{(k)}\) is appended once to the \(W\)-block IRS-service history in \eqref{eq:service_deficit}, irrespective of the recalibration mode. Thus, a recalibration block contributes one service sample using the assignment executed during its post-acquisition data interval.

Let \(R_t^{\mathrm{del},(k)}\) denote the realized net block rate obtained using the same payload accounting as \eqref{eq:net_rate}, with \(R_t^{\mathrm{del},(k)}=0\) for \(k\notin\mathcal D_t\). Before the first available delivered-rate sample,
\(\overline R_t^{(k)}=0\). Thereafter, the smoothed rate is updated as $\overline R_{t+1}^{(k)} = (1-\alpha_R)\overline R_t^{(k)} + \alpha_R R_t^{\mathrm{del},(k)}, \qquad 0<\alpha_R<1$. Once initialized, an inactive user's rate state decays through a zero delivered-rate observation.  Both the service and rate states are updated only after execution of the current action and therefore cannot introduce noncausal information into the slot-\(t\) decision.

Algorithm~\ref{alg:deterministic_controller} summarizes the complete online procedure. It is stated for \(\mathcal S_{N,t}^{\mathrm{feas}}\neq\varnothing\). If \(\mathcal S_{N,t}^{\mathrm{feas}}=\varnothing\), the implementation sets \(n_0=\min\mathcal S_N\), gives mandatory priority to active users with the largest residual service deficits, and applies the same deterministic ownership and beam-selection rules.

\begin{figure}[t!]
\algorithmheading{Full-Slot Uncertainty-Aware Structured Control}
\label{alg:deterministic_controller}
\begin{algorithmic}[1]
\Require  \(\mathcal D_t\); direct- and cascaded-channel posteriors; IRS-service history; elementary IRS grid; beam codebook; \(\mathcal S_N\); fixed calibration parameters.
\Ensure \(\mathcal A_t= (a_t^{\mathrm{cal}}, N_{0,t}^{\star}, \mathbf B_t^{\star}, \mathbf q_t^{\star}, \mathbf c_t^{\star})\).
\State Update the direct-channel posterior and form the predicted beam-domain cascaded states and covariances.
\If{an active user has no initialized cascaded state}
    \State \(a_t^{\mathrm{cal}}\gets1\).
\Else
    \State Compute \(U_{\mathrm{pre},t}\) and \(U_{\mathrm{floor},t}\), and select
    \(a_t^{\mathrm{cal}}\) using  \eqref{eq:actionable_recalibration}.
\EndIf
\If{\(a_t^{\mathrm{cal}}=1\)}
    \State Perform the differential repeated-data DFT acquisition and project the updated posterior onto the tracked unit-beam states.
\EndIf
\State Set \(\mathcal P_t^{\mathrm{ctrl}}\) to the recalibrated posterior if \(a_t^{\mathrm{cal}}=1\), and to the prediction-only posterior otherwise.
\State Construct \(\mathcal S_{N,t}^{\mathrm{feas}}\) using \eqref{eq:feasible_unit_sizes}.
\For{\(n_0\in\mathcal S_{N,t}^{\mathrm{feas}}\)}
   \State Determine the residual-feasible unit ownership according to Section~VII-B.
    \State Initialize the beams by \eqref{eq:isolated_beam_initialization} and perform one
    fixed-order sweep using \eqref{eq:coordinate_beam_refinement}.
    \State Evaluate \(\widehat{\mathcal R}_t (\widehat{\mathbf B}_t(n_0),\widehat{\mathbf q}_t(n_0))\)
\EndFor
\State Select \(N_{0,t}^{\star}\) using \eqref{eq:final_unit_size_selection}, and set
\(\mathbf B_t^{\star}\) and \(\mathbf q_t^{\star}\) to the corresponding candidate values.
\State Form the service-first user order and obtain \(\mathbf c_t^{\star}\) using the sequential allocator of Section~IX.
\State Execute \(\mathcal A_t\), compute the delivered net rates, and update the service and rate histories.
\State \Return \(\mathcal A_t\).
\end{algorithmic}
\algorithmendrule
\end{figure}

\section{Receiver-Consistent Channel Allocation}
\label{sec:sequential_channel_allocation}

After fixing the recalibration decision, IRS-unit ownership, and beam configuration, the controller assigns one uplink channel to each active user. Since uplink channels may be reused, each candidate assignment is evaluated using both its conditional uncertainty-aware ZF rate and the interference risk that it creates with the users already assigned to the same channel. To obtain a receiver-consistent pairwise risk metric independently of the sequential channel-assignment state, an allocation-independent reference ZF receiver is constructed from user-specific full-aperture reference beams.

\subsection{Receiver-Consistent Risk and Channel Metrics}

For \(q\in\mathcal Q_t^{\mathrm{tr}}\), define the full-aperture reflected response and its prediction-error covariance as $\widehat{\mathbf r}_{t+1|t}^{(k)}(q) \triangleq \sum_{u=1}^{U}\widehat{\mathbf r}_{u,t+1|t}^{(k)}(q)$  and $\mathbf P_{r,q,t+1|t}^{(k,k)} \approx \sum_{u=1}^{U} \mathbf P_{r,u,q,t+1|t}^{(k,k)}$, where the covariance approximation follows the cross-unit error-independence assumption of Section~IV-C. The corresponding reference effective channel and uncertainty are $\widehat{\mathbf h}_{t+1|t}^{\mathrm{ref},(k)}(q) \triangleq \widehat{\mathbf h}_{t+1|t}^{d,(k)} +\widehat{\mathbf r}_{t+1|t}^{(k)}(q)$, and $\eta_t^ {\mathrm{ref},(k)}(q) \triangleq \operatorname{Tr}\!\big( \mathbf P_{t+1|t}^{d,(k,k)} +\mathbf P_{r,q,t+1|t}^{(k,k)} \big)$. This aggregation is used only for the allocation-independent reference metric and does not replace the contiguous-block acquisition covariance used for the executable IRS candidates.

The user-specific reference beam is selected as
\begin{equation}
q_t^{\mathrm{ref},(k)}= \min\!\Big( \arg\max_{q\in\mathcal Q_t^{\mathrm{tr}}} \Big\{
\big\| \widehat{\mathbf h}_{t+1|t}^{\mathrm{ref},(k)}(q) \big\|_2^2 - \lambda_{\mathrm{ref}} \eta_t^{\mathrm{ref},(k)}(q) \Big\} \Big) \label{eq:reference_beam_selection}
\end{equation}
where \(\lambda_{\mathrm{ref}}\geq0\) controls the gain-uncertainty trade-off. For compactness, define  $\widehat{\mathbf h}_{t+1|t}^{\mathrm{ref},(k)}\triangleq \widehat{\mathbf h}_{t+1|t}^{\mathrm{ref},(k)} \big(q_t^{\mathrm{ref},(k)}\big)$ and $\mathbf P_{t+1|t}^{\mathrm{eff},(k,k),\mathrm{ref}} \triangleq \mathbf P_{t+1|t}^{d,(k,k)} + \mathbf P_{r,q_t^{\mathrm{ref},(k)},t+1|t}^{(k,k)}$. Stacking the reference effective channels of the active users gives $\widehat{\mathbf H}_{t+1|t}^{\mathrm{ref}} \triangleq \big[ \widehat{\mathbf h}_{t+1|t}^{\mathrm{ref},(k)} \big]_{k\in\mathcal D_t}$. The corresponding reference ZF combiner is $\mathbf W_{t+1|t}^{\mathrm{ref,ZF}} = \widehat{\mathbf H}_{t+1|t}^{\mathrm{ref}}\big[ \big( \widehat{\mathbf H}_{t+1|t}^{\mathrm{ref}} \big)^H \widehat{\mathbf H}_{t+1|t}^{\mathrm{ref}} \big]^{\dagger}$, where $\mathbf w_{t+1|t}^{\mathrm{ref},(i)}$ denotes the column associated with user $i$. 

For $i,j\in\mathcal D_t$ and $i\neq j$, define $s_{t+1|t}^{\mathrm{ref},(j\rightarrow i)} \triangleq\big|(\mathbf w_{t+1|t}^{\mathrm{ref},(i)})^{H} \widehat{\mathbf h}_{t+1|t}^{\mathrm{ref}, (j)}\big|^2$. The directed receiver-consistent interference coupling from user $j$ to user $i$ is
 \begin{equation}
e_{\mathrm{int}, t}^{ (j\rightarrow i)} = \frac{ p^{(j)} \Big[ s_{t+1| t}^{\mathrm{ref},(j\rightarrow i)} + (\mathbf w_{t+1|t}^{\mathrm{ref},(i)})^{H} \mathbf P_{t+1| t}^{\mathrm{eff}, (j,j), \mathrm{ref}} \mathbf w_{t+1|t}^{\mathrm{ref}, (i)} \Big] }{ \sigma_n^2 \left\| \mathbf w_{t+1|t}^{\mathrm{ref}, (i)} \right\|_2^2 +\epsilon_I }. \label{eq:directed_receiver_consistent_risk}
\end{equation}
where $\epsilon_I>0$ is a fixed numerical safeguard with the same scale as the post-combining noise power. This metric is directional because the receive combiner associated with user $i$ determines the interference produced by user $j$. It accounts for both the predicted residual multiuser coupling and the interference induced by channel-prediction uncertainty under the same ZF receive-processing model used by the controller. For full-column-rank reference channels, the deterministic cross-user term vanishes under exact ZF; it is retained for pseudoinverse or rank-deficient cases.

For \(c\in\mathcal C\), let \(\mathcal{K}_{t,c}\subseteq\mathcal{D}_t\) denote the users already assigned to channel \(c\), with \(\mathcal K_{t,c}=\varnothing\) before the first assignment. For a candidate user \(k\), define the bidirectional receiver-consistent cochannel risk as
\begin{equation}
\Omega_{t,c}^{(k)} = \begin{cases} 0, & \mathcal{K}_{t,c}=\varnothing, \\[2mm]
\displaystyle \frac{1}{|\mathcal{K}_{t,c}|} \sum_{j\in\mathcal{K}_{t,c}} \big(
e_{\mathrm{int},t}^{(j\rightarrow k)} + e_{\mathrm{int},t}^{(k\rightarrow j)} \big),
& \mathcal{K}_{t,c}\neq\varnothing. \end{cases} \label{eq:bidirectional_channel_risk}
\end{equation}
Both directions are retained because receive combining generally makes cochannel coupling asymmetric: $e_{\mathrm{int},t}^{(j\rightarrow k)}$ measures the risk imposed on candidate user $k$, whereas $e_{\mathrm{int},t}^{(k\rightarrow j)}$ measures the risk introduced by $k$ to an already assigned user $j$. Averaging over $\mathcal K_{t,c}$ limits an additional occupancy bias, since channel load is already reflected in the provisional ZF rate.

For candidate pair \((k,c)\), define the provisional cochannel set \(\mathcal S_{t,c}^{(k)}=\mathcal K_{t,c}\cup\{k\}\). Under the selected IRS action \((\mathbf B_t^\star,\mathbf q_t^\star)\), form the provisional predicted-channel matrix as
\begin{equation}
\widehat{\mathbf H}_{t,c}^{(k)} = \big[ \widehat{\mathbf h}_{t+1|t}^{\mathrm{eff},(j),\mathrm{ctrl}} \big( \mathbf B_t^{\star}, \mathbf q_t^{\star}; a_t^{\mathrm{cal}}
\big) \big]_{j\in\mathcal{S}_{t,c}^{(k)}}, \label{eq:provisional_effective_matrix}
\end{equation}
where the columns are arranged in increasing user-index order. The provisional ZF
combiner is
\begin{equation}
\mathbf W_{t,c}^{\mathrm{ZF},(k)}=\widehat{\mathbf H}_{t,c}^{(k)}\big[
\big(\widehat{\mathbf H}_{t,c}^{(k)}\big)^H \widehat{\mathbf H}_{t,c}^{(k)} \big]^{\dagger}, \label{eq:provisional_zf_combiner}
\end{equation}
where \((\cdot)^{\dagger}\) denotes the Moore-Penrose pseudoinverse. Let \(\mathbf w_{t,c}^{(k)}\) denote the column of \eqref{eq:provisional_zf_combiner} associated with user \(k\). 

The conditional uncertainty-aware SINR is
\begin{equation}
\widehat{\gamma}_{t,c}^{(k)} = \frac{ p^{(k)}  \big|\big(\mathbf w_{t,c}^{(k)}\big)^H
\widehat{\mathbf h}_{t+1|t}^{\mathrm{eff}, (k),\mathrm{ctrl}} \big|^2 }{
\mathcal{I}_{t,c}^{(k)} }, \label{eq:conditional_ua_zf_sinr}
\end{equation}
where
\begin{align}
\mathcal I_{t,c}^{(k)}={}&\sigma_n^2\big\|\mathbf w_{t,c}^{(k)}\big\|_2^2 +
p^{(k)} \big(\mathbf w_{t,c}^{(k)}\big)^H\mathbf P_{t+1|t}^{\mathrm{eff}, (k,k),\mathrm{ctrl}} \mathbf w_{t,c}^{(k)} \nonumber\\
&+ \sum_{j\in \mathcal S_{t,c}^{(k)}\setminus\{k\}} p^{(j)} \Big[ \big| (\mathbf w_{t,c}^{(k)})^H \widehat{\mathbf h}_{t+1|t}^{\mathrm{eff},(j),\mathrm{ctrl}} \big|^2
\nonumber\\ 
&\hspace{18mm} + (\mathbf w_{t,c}^{(k)})^H \mathbf P_{t+1|t}^{\mathrm{eff},(j,j),\mathrm{ctrl}} \mathbf w_{t,c}^{(k)} \Big]. \label{eq:conditional_ua_zf_interference}
\end{align}
 These expressions specialize the uncertainty-aware SINR of Section~VI to the provisional cochannel set. The channel-dependent rate is
\begin{equation}
\widehat R_{t,c}^{(k)} = \zeta_t^{\mathrm{ctrl}} B_{\mathrm{ch}} \log_2\!\left(1+\widehat{\gamma}_{t,c}^{(k)}\right), \label{eq:conditional_controlled_rate}
\end{equation}
where \(\zeta_t^{\mathrm{ctrl}}\triangleq (1-a_t^{\mathrm{cal}})\frac{S-\tau_p}{S} +
a_t^{\mathrm{cal}}(\frac{\beta}{S})\). Equation~\eqref{eq:conditional_controlled_rate} retains the channel-dependent component of the net block rate in \eqref{eq:net_rate}. In a recalibration block, the fresh minimum-reflection baseline contribution is identical for all candidate channels and is therefore omitted from channel ranking while remaining included in the delivered block rate.

The rate and cochannel-risk terms are normalized using the clipped training-statistics mapping in \eqref{eq:7}. For the candidate rate $\widehat{R}^{(k)}_{t,c}$, we use $(x_{\min}^{\rm tr},x_{\max}^{\rm tr},\epsilon_x)=(R_{\min}^{\rm tr},R_{\max}^{\rm tr},\epsilon_R)$ and denote the result by $\mathcal{N}_R(\widehat{R}^{(k)}_{t,c})$, where \(R_{\min}^{\mathrm{tr}}\) and \(R_{\max}^{\mathrm{tr}}\) are the frozen bounds. For the interference risk, define $\widetilde{\Omega}_{t,c}^{(k)} = \frac{\Omega_{t,c} ^{(k)}} {s_{\rm int}^{\rm tr}}$, where $s_{\rm int}^{\rm tr}>0$ is the fixed training-set scale, and apply $\mathcal N_\Omega$ using $(0,\Omega_{\max}^{\rm tr},\epsilon_\Omega)$. All normalization statistics remain fixed after training.

\subsection{Service-First Sequential Allocation}

Users are processed in service-first order because each channel utility depends on the cochannel sets formed by earlier assignments. For $k\in\mathcal D_t$, define the normalized service-priority score $F_t^{(k)} \triangleq \min\!\left\{ 1,\frac{d_t^{(k)}}{s_{\min}^{(k)} +\varepsilon_F} \right\}$, where $\varepsilon_F>0$ and $d_t^{(k)}$ is the residual aperture deficit in \eqref{eq:service_deficit}. Users with equal service priority are ordered by their singleton predicted rate  $\widehat R_{t,\mathrm{iso}}^{(k)} = \zeta_t^{\mathrm{ctrl}} B_{\mathrm{ch}} \log_2 \big( 1+ \widehat{\gamma}_{t,\mathrm{iso}}^{(k)} \big)$, where \(\widehat{\gamma}_{t,\mathrm{iso}}^{(k)}\) follows from \eqref{eq:conditional_ua_zf_sinr} with \(\mathcal S_{t,c}^{(k)}=\{k\}\). 

Let $\boldsymbol{\pi}_t = [\pi_t(1),\ldots,\pi_t(|\mathcal{D}_t|)]^T$ denote the resulting user order. It satisfies $\pi_t(m)\prec_t\pi_t(n)$ for $m<n$, where
\begin{equation}
k\prec_t j \Longleftrightarrow \begin{cases} F_t^{(k)}>F_t^{(j)}, & \text{or}\\[1mm]
F_t^{(k)}=F_t^{(j)} \ \text{and}\ \widehat R_{t,\mathrm{iso}}^{(k)} >
\widehat R_{t,\mathrm{iso}}^{(j)}. \end{cases}
\label{eq:final_channel_processing_order}
\end{equation}
Any remaining tie is resolved by the smaller user index. For the current user $k=\pi_t(m)$ and candidate channel $c$, the allocator evaluates $ \mathcal U_{t,c}^{(k)} = \mathcal N_R\!\big(\widehat R_{t,c}^{(k)}\big)- \lambda_G  \mathcal N_\Omega\!\big( \widetilde{\Omega}_{t,c}^{(k)}\big)$, where $\lambda_G>0$ is fixed during controller development. The channel is selected deterministically as $c_t^{(k)} = \min\!\big( \arg\max_{c\in\mathcal C} \mathcal U_{t,c}^{(k)} \big)$, after which $\mathcal K_{t,c_t^{(k)}} \leftarrow \mathcal K_{t,c_t^{(k)}}\cup\{k\}$.

After all active users have been processed, the final allocation is $\mathbf c_t^\star =[c_t^{(k)}]_{k\in\mathcal D_t}$. No separate load penalty is used because channel occupancy is already reflected in the provisional ZF rate, and Section~XII-B shows no repeatable gain from an additional load term. The complete deterministic procedure is summarized in Algorithm~\ref{alg:sequential_channel_allocator}. Candidate rates are recomputed after each assignment, so all previous decisions are reflected in the subsequent provisional ZF receivers. The procedure therefore evaluates \(C|\mathcal D_t|\) channel candidates instead of enumerating \(C^{|\mathcal D_t|}\) complete channel vectors.

\begin{figure}[t]
\algorithmheading{Receiver-Consistent Sequential Uplink-Channel Allocation}
\label{alg:sequential_channel_allocator}
\begin{algorithmic}[1]
\Require \(\mathcal D_t\); $a_t^{\mathrm{cal}}$; \((\mathbf B_t^\star,\mathbf q_t^\star)\); control-conditioned effective-channel means and covariances; \(\{F_t^{(k)}\}\); \(\{e_{\mathrm{int},t}^{(j\rightarrow k)}\}\); \(\mathcal C\); $\lambda_G$.
\Ensure Uplink-channel vector \(\mathbf c_t^{\star}\).
\State Set \(\mathcal K_{t,c}\gets\varnothing\), \(c\in\mathcal C\).
\State Compute \(\{\widehat R_{t,\mathrm{iso}}^{(k)}\}_{k\in\mathcal{D}_t}\),
sort the active users according to \eqref{eq:final_channel_processing_order}, and obtain \(\boldsymbol{\pi}_t\).
\For{\(m=1,\ldots,|\mathcal{D}_t|\)}
    \State Set \(k\leftarrow\pi_t(m)\).
    \For{each \(c\in\mathcal{C}\)}
       \State Form \(\mathcal{S}_{t,c}^{(k)} =\mathcal{K}_{t,c}\cup\{k\}\).
        \State Recompute the provisional ZF combiner by \eqref{eq:provisional_zf_combiner}.
   \State Compute \(\widehat R_{t,c}^{(k)}\) and \(\Omega_{t,c}^{(k)}\) using
        \eqref{eq:conditional_ua_zf_sinr}-\eqref{eq:conditional_controlled_rate} and
        \eqref{eq:bidirectional_channel_risk}.
       \State Compute \( \mathcal U_{t,c}^{(k)} \gets \mathcal N_R(\widehat R_{t,c}^{(k)}) - \lambda_G \mathcal N_\Omega(\widetilde{\Omega}_{t,c}^{(k)})\).
    \EndFor
    \State Select \(c_t^{(k)} \gets \min\!\big( \arg\max_{c\in\mathcal C} \mathcal U_{t,c}^{(k)} \big)\).
    \State Update \(\mathcal K_{t,c_t^{(k)}} \gets \mathcal K_{t,c_t^{(k)}}\cup\{k\}\).
\EndFor
\State Form \(\mathbf c_t^\star=[c_t^{(k)}]_{k\in\mathcal D_t}\) in increasing user-index order. 
\State \textbf{return} \(\mathbf c_t^{\star}\).
\end{algorithmic}
\algorithmendrule
\end{figure}

\section{Theoretical Properties}

This section analyzes IRS-service feasibility, prediction-error bound, reliability, and recalibration rule properties of the deterministic controller developed in Sections~VII–IX.

\subsection{IRS Feasibility and Effective-Channel Error Bounds}

\begin{proposition}
\label{prop:irs_feasibility}
For \(n_0\in\mathcal S_N\), an exclusive assignment of the \(G(n_0)\) equal-size IRS units satisfying the current residual service requirements exists if and only if $\sum_{k\in\mathcal D_t}n_{\mathrm{req},t}^{(k)}(n_0) \leq G(n_0)$. Moreover, if \(n_0\in\mathcal S_{N,t}^{\mathrm{feas}}\) and, at each step \(g\), the owner is selected from
\(\mathcal K_{t,g}^{\mathrm{feas}}(n_0)\) in  \eqref{eq:feasible_owner_set}, then
\begin{equation}
\begin{gathered}
\sum_{j\in\mathcal D_t} \big[ n_{\mathrm{req},t}^{(j)}(n_0) - n_{t,g}^{(j)}(n_0)
\big]^+ \leq G(n_0)-g, \\
g=1,\ldots,G(n_0),
\end{gathered}\label{eq:residual_feasibility_invariant}
\end{equation}
where $n_{t,g}^{(j)}(n_0) = \sum_{\ell=1}^{g} \mathbf 1\!\big\{ k_{\ell,t}^{\star}(n_0)=j \big\}$. Consequently, \(n_{t,G(n_0)}^{(k)}(n_0) \geq n_{\mathrm{req},t}^{(k)}(n_0)\) for every \(k\in\mathcal D_t\).
\end{proposition}

\begin{IEEEproof}
Necessity follows because user $k$ requires at least $n_{\mathrm{req},t}^{(k)}(n_0)$ distinct units. Conversely, if the total requirement does not exceed $G(n_0)$, the required units can be reserved first and the remaining units assigned arbitrarily. For the sequential rule, suppose that before processing unit $g$ the remaining demand does not exceed $G(n_0)-g+1$. If this demand is positive, assigning the current unit to any user with positive residual demand reduces it by one; if it is zero, every active user is feasible. Hence, $\mathcal K_{t,g}^{\mathrm{feas}}(n_0)\neq\varnothing$, and \eqref{eq:residual_feasibility_invariant} is preserved. Induction over $g=1,\ldots,G(n_0)$ yields zero residual demand after the final assignment, proving the claim.
\end{IEEEproof}
Thus, for every \(n_0\in\mathcal S_{N,t}^{\mathrm{feas}}\), the owner selection in \eqref{eq:deterministic_owner} preserves the current residual service requirements without
post-processing repair. Let \(\mathcal I_t\) denote the information available after the slot-\(t\) control posterior is formed and before the IRS configuration is selected. It contains the recalibrated posterior when \(a_t^{\mathrm{cal}}=1\) and the prediction-only posterior otherwise. For the selected configuration \((N_{0,t}^{\star},\mathbf B_t^\star,\mathbf q_t^\star)\), define \(G_t^\star\triangleq G(N_{0,t}^\star)\) and \(\mathcal N_{g,t}^\star\triangleq \mathcal N_g(N_{0,t}^\star)\). Let $\Delta\mathbf L_{t+1|t}^{g, (k)} =\big[\mathbf L_{t+1}^{(k)} -\widehat{\mathbf L}_{t+1|t}^{(k)} \big]_{:,\mathcal{N}^{\star}_{g,t}}$ denote the cascaded-channel prediction error over IRS unit \(g\), and  all selected units have the same size, $|\mathcal{N}^{\star}_{g,t}|=N_{0,t}^{\star}$ for $g=1,\ldots,G_t^{\star}$.

\begin{theorem}
Assume that, conditioned on \(\mathcal{I}_t\), $\mathbb{E} \big[\big\| \Delta\mathbf h_{t+1|t}^{d,(k)} \big\|_2^2 \bigm| \mathcal{I}_t \big] \leq \xi_d^{(k)}$ and $\mathbb{E} \big[ \big\| \Delta\mathbf L_{t+1|t}^{g,(k)} \big\|_F^2 \bigm| \mathcal{I}_t \big] \leq \xi_{L,g}^{(k)}$ for every active user \(k\) and selected IRS unit \(g\). Then, under the intended-user dominant-reflection model and the unit-modulus codewords of Section~V, the controller-side effective-channel prediction error satisfies
\begin{equation}
\mathbb{E}\!\Big[ \big\|\Delta\mathbf h_{t+1|t}^{\mathrm{eff}, (k)}\big\|_2^2 \bigm|\mathcal I_t \Big] \leq \Big( \sqrt{\xi_d^{(k)}} + \sum_{g=1}^{G_t^\star}
[\mathbf B_t^\star]_{g,k} \sqrt{N_{0,t}^\star\xi_{L,g}^{(k)}} \Big)^2.
\label{eq:effective_channel_mse_bound}
\end{equation}
\end{theorem}
\begin{IEEEproof} Conditioned on \(\mathcal I_t\), the deterministic action \((\mathbf B_t^\star, \mathbf q_t^\star)\) is fixed, and the effective-channel prediction error is
\begin{equation}
\Delta\mathbf h_{t+1|t}^{\mathrm{eff},(k)} = \Delta\mathbf h_{t+1|t}^{d,(k)}
+ \sum_{g=1}^{G_t^\star} [\mathbf B_t^\star]_{g,k} \Delta\mathbf L_{t+1|t}^{g, (k)} \mathbf v_{g,q_{g,t}^{\star}}. \label{eq:effective_error_decomposition}
\end{equation}
Minkowski's inequality, together with \(\|\mathbf A\mathbf x\|_2 \leq\|\mathbf A\|_F\|\mathbf x\|_2\) and $\big\| \mathbf v_{g,q_{g,t}^{\star}} \big\|_2^2 = N_{0,t}^{\star}$, gives
\begin{align}
&\sqrt{ \mathbb{E}\Big[\big\|\Delta\mathbf h_{t+1|t}^{\mathrm{eff}, (k)}\big\|_2^2
\bigm| \mathcal{I}_t \Big]}\nonumber\\
&\quad\leq \sqrt{\varepsilon_d^{(k)}} + \sum_{g=1}^{G^{\star}_t} a_{g,t}^{(k)}\sqrt{
\mathbb{E} \Big[ \big\|\Delta\mathbf L_{t+1|t}^{g,(k)}\mathbf v_{g,q_{g,t}^{\star}}
\big\|_2^2 \bigm| \mathcal{I}_t \Big] }, \label{eq:minkowski_effective_error}
\end{align}
which proves \eqref{eq:effective_channel_mse_bound}.
\end{IEEEproof}
The bound does not require independence among the selected control-unit prediction errors.

\subsection{Reliability Bounds and Recalibration Actionability}

For the selected IRS action, define $\mathbf e_t^{(k)} \triangleq \mathbf h_{t+1}^{\mathrm{eff},(k)} (\mathbf B_t^\star,\mathbf q_t^\star) - \widehat{\mathbf h}_{t+1|t}^{\mathrm{eff},(k),\mathrm{ctrl}} (\mathbf B_t^\star,\mathbf q_t^\star;a_t^{\mathrm{cal}})$,  where the selected-action arguments are omitted for compactness.  Assume that the control-conditioned uncertainty model upper-bounds the conditional effective-channel mean-square error as $ \mathbb{E}\!\big[ \big\|\mathbf e_t^{(k)}\big\|_2^2 \bigm|\mathcal I_t \big] \leq \eta_t^{(k)}$ and $\eta_t^{(k)}
\triangleq \operatorname{Tr}\!\big( \mathbf P_{t+1|t}^{\mathrm{eff},(k,k),\mathrm{ctrl}}
(\mathbf B_t^\star,\mathbf q_t^\star;a_t^{\mathrm{cal}}) \big)$.
\begin{cor}
For every \(\delta_k\in(0,1)\),
\begin{equation}
\Pr\!\Big(\big\|\mathbf e_t^{(k)}\big\|_2 \leq \sqrt{\frac{\eta_t^{(k)}}{\delta_k}} | \mathcal I_t \Big) \geq 1-\delta_k. \label{eq:high_probability_effective_error_bound}
\end{equation}
\end{cor}
\begin{IEEEproof}
By assumption, $\mathbb{E}\!\big[ \|\mathbf e_t^{(k)}\|_2^2 \bigm|\mathcal I_t
\big] \leq \eta_t^{(k)}$. Conditional Markov's inequality applied to \(\|\mathbf e_t^{(k)}\|_2^2\) yields $\Pr\!\big( \|\mathbf e_t^{(k)}\|_2^2 > \frac{\eta_t^{(k)}}{\delta_k} |\mathcal I_t \big) \leq \delta_k$, which proves \eqref{eq:high_probability_effective_error_bound}.
\end{IEEEproof}

The result is distribution free once the stated trace-dominance condition holds.  Hence, a smaller \(\eta_t^{(k)}\) yields a tighter conditional error bound and gives this covariance trace a direct reliability interpretation. Before IRS selection, the controller computes \(U_{\mathrm{pre},t}\) and \(U_{\mathrm{floor},t}\) from \eqref{eq:actionable_recalibration} to determine whether recalibration is actionable. When all active users have initialized cascaded-channel states, recalibration is selected only if \(U_{\mathrm{pre},t}>\varepsilon_{\mathrm{cal}}\) and \(U_{\mathrm{floor},t}<U_{\mathrm{pre},t}\).  Hence, \(U_{\mathrm{floor},t}\geq U_{\mathrm{pre},t}\) suppresses recalibration even when the absolute-risk threshold is  exceeded.  The first condition flags excessive normalized prediction risk, whereas the second requires a predicted reduction in the worst-user normalized risk. The actionability test is model based and does not guarantee a smaller realized estimation error in every block.

\section{Computational Complexity Analysis}

This section characterizes the online complexity of the channel processing in Sections~III-IV and the structured controller in Sections~VII-IX. Offline transition-model identification, covariance calibration and lookup tables, and training-set normalization statistics are excluded. Let \(a_t\triangleq|\mathcal D_t|\), \(Q_t^{\mathrm{tr}}\triangleq|\mathcal Q_t^{\mathrm{tr}}|\), and \(Q_t^{\mathrm{cand}}\triangleq|\mathcal Q_t^{\mathrm{cand}}|\). Let \(K_t^{\mathrm{cas}}\leq K\) denote the number of users with initialized cascaded-channel states, and let \(U\) be the number of elementary IRS units.  Define $\overline{G}_t \triangleq \sum_{n_0\in\mathcal{S}_{N,t} ^{\mathrm{feas}}} G(n_0)$. For the fallback case \(\mathcal S_{N,t}^{\mathrm{feas}}=\varnothing\), set \(\overline G_t=G(\min\mathcal S_N)\), consistent with Section~VIII-B.

For a provisional ZF receiver with \(s\) predicted effective channels, the Gram-matrix construction and pseudoinverse require $\mathscr C_{\mathrm{ZF}}(s) = \mathcal O(Ms^2+s^3)$.  With dense \(M\times M\) effective-channel error covariances, one uncertainty-aware ZF rate evaluation has complexity $\mathscr C_{\mathrm{UAZF}}^{(1)}(s) = \mathcal O(Ms^2+s^3+sM^2)$, while evaluating all \(s\) rates under the same configuration is upper-bounded by $\mathscr C_{\mathrm{UAZF}}^{(\Sigma)}(s) = \mathcal O(Ms^2+s^3+s^2M^2)$.

\subsection{Channel Processing and Physical Metrics}

The covariance, cross-covariance, and gain matrices of the O-CA tracker have complexity $\mathscr{C}_{t}^{\mathrm{OCA}} = \mathcal{O}\big(\frac{a_{t}(a_{t}+1)}{2}M^3\big)$. Under the steady-state conditions of Remark~2, these matrices can be computed offline, and the online prediction and correction of the direct-channel means require $\mathscr{C}_{t}^{\mathrm{OCA,ss}}= \mathcal{O} \big( (K+a_t)M^2 \big)$. When \(a_t^{\mathrm{cal}}=1\), semi-blind acquisition processes one minimum-reflection block and \(N\) DFT-coded repeated-data blocks. The differential recovery, inverse DFT, and projection onto the tracked unit-beam states yield the recalibration-stage complexity $\mathscr C_t^{\mathrm{SB}} = \mathcal O\!\big(NM^2 (\alpha+a_t) +a_tMN(\log N+Q_t^{\mathrm{tr}}) \big)$, where the lower-dimensional \(a_t\times a_t\) operations are absorbed by this bound since \(a_t\leq M\).
      
For comparison, covariance propagation of the full \(MN\)-dimensional cascaded state would require $\mathcal O\!\big(K_t^{\mathrm{cas}}M^3N^3\big)$, whereas propagation of all \(M\)-dimensional elementary unit-beam covariances requires  $\mathcal O\!\big( K_t^{\mathrm{cas}}UQ_t^{\mathrm{tr}}M^3 \big)$. In the structured implementation of Section~IV-C, the covariance sequences are retrieved from precomputed age-dependent tables, so only the unit-beam means are propagated online, yielding $\mathscr C_t^{\mathrm{beam}} = \mathcal O\!\big( K_t^{\mathrm{cas}}UQ_t^{\mathrm{tr}}M^2 \big)$. The actionability test requires $\mathscr C_t^{\mathrm{cal}} = \mathcal O(a_tUQ_t ^{\mathrm{tr}}M)$. The full-aperture reference responses are accumulated during the same traversal and cached. Hence, reference-beam evaluation adds \(\mathcal O(a_tQ_t^{\mathrm{tr}}M)\), while forming the reference ZF receiver and all directed receiver-consistent interference metrics requires $\mathscr C_t^{\mathrm{int}} = \mathcal O\!\big( \mathscr C_{\mathrm{ZF}}(a_t)+a_t^2M^2 \big)$.

\subsection{Structured Resource-Control Complexity}

Across all feasible unit sizes, the isolated ownership stage requires \(a_tQ_t^{\mathrm{cand}}\overline G_t\) system-rate evaluations, while the coordinate sweep requires an additional \(Q_t^{\mathrm{cand}}\overline G_t\). Therefore, the dominant IRS-side physical evaluation cost is $\mathscr C_t^{\mathrm{IRS}} = \mathcal O\!\left( Q_t^{\mathrm{cand}}(a_t+1)\overline G_t \mathscr C_{\mathrm{UAZF}} ^{(\Sigma)}(a_t) \right)$. 

The candidate-size feasibility tests and incremental residual-feasibility updates add only $\mathcal O\!\big( |\mathcal S_N|a_t+a_t\overline G_t \big)$ operations.
After fixing the IRS action, the singleton rates used for tie breaking are computed once and cached; their cost is dominated by the subsequent channel-candidate evaluations.  Sorting the active users requires \(\mathcal O(a_t\log a_t)\). Let \(s_{t,c}^{(m)}\) denote the provisional cochannel-set size when the \(m\)-th processed user is evaluated on channel \(c\). The sequential channel-allocation cost is $\mathscr C_t^{\mathrm{ch}} = \mathcal O\!\big( \sum_{m=1}^{a_t}\sum_{c=1}^{C} \mathscr C_{\mathrm{UAZF}}^{(1)} \!\big(s_{t,c}^{(m)}\big) +a_t^2+a_t\log a_t \big)$, where the \(a_t^2\) term accounts for accumulated pairwise-risk summations using the cached receiver-consistent metrics. Since \(s_{t,c}^{(m)}\leq a_t\), a simpler worst-case bound is $\mathscr{C}_{t}^{\mathrm{ch}} = \mathcal{O} \left( Ca_t \left[ Ma_t^2+a_t^3+a_tM^2
\right] + a_t^2 + a_t\log a_t \right)$.

\subsection{Overall Online Complexity and Scaling}

With per-user rolling sums, the service-history and smoothed-rate updates require \(\mathcal O(K)\) operations per block. Omitting the lower-order feasibility, ordering, and cached reference-beam terms, a compact online upper bound is
\begin{align}
\mathscr{C}_{t}^{\mathrm{online}}=\mathcal{O}&\Big(\frac{a_{t}(a_{t}+1)}{2}M^3+K_t^{\mathrm{cas}} UQ_t^{\mathrm{tr}}M^2 +a_tUQ_t^{\mathrm{tr}}M\nonumber\\
& + a_t^{\mathrm{cal}}\big[NM^2(\alpha+a_t)+a_tMN\big(\log N+Q_t^{\mathrm{tr}}
\big) \big] \nonumber\\
& + \mathscr{C}_{\mathrm{ZF}}(a_t) + a_t^2M^2 + Ca_t\mathscr C_{\mathrm{UAZF}}^{(1)}(a_t) +K \nonumber\\ 
& + Q_t^{\mathrm{cand}}(a_t+1) \overline{G}_t \mathscr{C}_{\mathrm{UAZF}}^{(\Sigma)}(a_t) \Big), \label{eq:overall_online_complexity}
\end{align}
where the semi-blind term is incurred only when \(a_t^{\mathrm{cal}}=1\).

Conditioned on the selected recalibration mode, an unconstrained enumeration of owner, beam, and channel choices contains at most $N_t^{\mathrm{ex}} = C^{a_t}
\sum_{n_0\in\mathcal S_{N,t}^{\mathrm{feas}}} \left(a_tQ_t^{\mathrm{cand}} \right)^{G(n_0)}$ candidates. In contrast, the structured procedure evaluates at most $N_t^{\mathrm{IRS,eval}} = Q_t^{\mathrm{cand}}(a_t+1)\overline G_t$ IRS candidates and $N_t^{\mathrm{ch,eval}}=Ca_t$ sequential channel candidates. Thus, for fixed \(|\mathcal S_N|\), \(Q_t^{\mathrm{cand}}\), and \(C\), the resource-control search scales polynomially rather than exponentially with the joint action dimension. Reduced-order prediction replaces full $MN$-dimensional covariance propagation with independent $M$-dimensional unit-beam mean propagations and offline covariance retrieval.

\section{Simulation Setup and Numerical Results}

This section evaluates the proposed channel-acquisition, prediction, and deterministic resource-control framework. We consider the single-cell uplink model of Section~II, with the main physical, mobility, IRS, service, and controller parameters summarized in Table~\ref{tab:simulation_parameters}. 

\subsection{Simulation Setup and Holdout Summary}

The BS and IRS centers are located at $(0,0,10)$~m and $(50,20,5)$~m, respectively, with users at a height of $1.5$~m. Initial user positions are independently drawn over $10\leq x\leq100$~m and $-40\leq y\leq60$~m, subject to a $10$-m minimum distance from both the BS and IRS. Users move with random initial directions and speeds up to the nominal $v_{\max}=3$~m/s, with boundary reflections used to confine the trajectories. The direct link follows the adopted urban-microcell model, whereas the IRS-related links include geometry-dependent large-scale attenuation and Rician small-scale fading. Empty activity realizations are avoided by activating one user.

\begin{table}[t]
\centering
\caption{Main Simulation Parameters}
\label{tab:simulation_parameters}
\renewcommand{\arraystretch}{1.08}
\begin{tabular}{l c}
\hline
\textbf{Parameter} & \textbf{Value} \\
\hline
BS antennas, \(M\) & \(100\) \\
Users in the pilot group, \(K\) & \(5\) \\
Carrier frequency & \(3.5\) GHz \\
Subcarrier spacing & \(30\) kHz \\
Active subcarriers & \(96\) \\
Uplink-channel bandwidth, \(B_{\mathrm{ch}}\) & \(2.88\) MHz \\
Physical control-block duration & \(2\) ms \\
OFDM symbols per control block & \(56\) \\
User transmit power, \(p^{(k)}\) & \(23\) dBm \\
Receiver noise figure & \(5\) dB \\
Maximum nominal user speed & \(3\) m/s \\
User activity probability & \(0.8\) \\
IRS control dimension, \(N\) & \(32\) \\
Elementary IRS resolution, \(\Delta N\) & \(8\) \\
Beam-codeword count & \(4\) \\
IRS Rician factor & \(10\) dB \\
Candidate unit sizes, \(\mathcal S_N\) & \(\{8,16,32\}\) \\
Service window, \(W\) & \(4\) \\
Minimum \(W\)-block service, \(s_{\min}^{(k)}\) & \(8\) elements, \(\forall k\) \\
Rate-smoothing factor, \(\alpha_R\) & \(0.05\) \\
Recalibration threshold, \(\varepsilon_{\mathrm{cal}}^{\mathrm{dB}}\)
& \(-3.5\) dB \\
Channel-risk coefficient, \(\lambda_G\) & \(2\) \\
Reference-beam uncertainty coefficient, $\lambda_{\mathrm{ref}}$ & $1$ \\
Uplink-channel counts & \(C\in\{2,3\}\) \\
\hline
\end{tabular}
\end{table}

For repeated-data acquisition, one OFDM symbol over the $96$ active subcarriers is used for direct-channel pilots, giving $\tau_p=96$ channel uses, while the repeated unknown block has $\alpha=96$ channel uses. With $N=32$, recalibration uses one minimum-reflection baseline and $32$ DFT-coded IRS states. Each physical block contains $S=56\times96=5376$ channel uses, leaving $\beta=S-\tau_p-\alpha(N+1)=2112$ regular-data channel uses, equivalent to $22$ OFDM symbols. Since only the first repeated-data copy carries fresh payload, the fresh-payload fractions are $23/56=41.071\%$ in recalibration blocks and $55/56$ in prediction-only blocks. These factors are included directly in all reported net block rates.

The actionability-aware recalibration rule of Section~VII-A is compared with the absolute-risk reference $a_{t,\mathrm{abs}}^{\mathrm{cal}} =\mathbf{1}\{U_{\mathrm{pre},t} >\epsilon_{\mathrm{cal}}\}$. The proposed rule additionally requires the predicted post-acquisition covariance floor to be lower than the current prediction risk. Both rules are evaluated on paired trajectories over $30$ previously unseen environments, each containing $160$ physical blocks, with no holdout-based parameter selection.

\begin{table}[t]
\centering
\caption{Final Holdout Comparison of Recalibration Rules}
\label{tab:actionability_holdout}
\renewcommand{\arraystretch}{1.08}
\begin{tabular}{l cc}
\hline
\textbf{Metric} &
\textbf{Absolute} &
\textbf{Actionable} \\
\hline
Mean recalibration rate [\%] & 3.103 & 2.126 \\
90th percentile [\%] & 6.897 & 5.172 \\
95th percentile [\%] & 12.931 & 6.034 \\
Maximum environment rate [\%] & 17.241 & 7.759 \\
Maximum consecutive blocks & 7 & 2 \\
\(C=2\) realized rate [Mbps] & 121.734 & 122.492 \\
\(C=3\) realized rate [Mbps] & 155.535 & 156.667 \\
\(C=2\) Jain index & 0.9060 & 0.9041 \\
\(C=3\) Jain index & 0.9405 & 0.9396 \\
\(C=2\) min. smoothed rate [Mbps] & 16.535 & 16.680 \\
\(C=3\) min. smoothed rate [Mbps] & 22.509 & 22.716 \\
\hline
\end{tabular}
\end{table}

As summarized in Table~\ref{tab:actionability_holdout}, the actionability condition reduces the mean recalibration rate from $3.103\%$ to $2.126\%$. The $95$th-percentile and maximum environment-level rates decrease from $12.931\%$ to $6.034\%$ and from $17.241\%$ to $7.759\%$, respectively, while the longest recalibration sequence decreases from seven to two blocks. Of the $150$ blocks satisfying the absolute-risk condition, $75$ ($50\%$) are rejected by the actionability test. This reduction does not degrade realized throughput: the mean sum rate increases by $0.622\%$ for $C=2$ and $0.728\%$ for $C=3$. The corresponding Jain-index changes are only $-0.0020$ and $-0.0009$, while the minimum smoothed-user rates improve by $0.878\%$ and $0.920\%$. The zero median paired throughput change for both channel counts further indicates that the proposed rule mainly suppresses recalibrations whose predicted covariance floor provides no uncertainty reduction. Thus, the actionability condition acts as a structural safeguard rather than an additional continuously tuned control term.

\subsection{Numerical Results and Design Validation}

Having fixed the simulation configuration and holdout protocol, we next evaluate the acquisition accuracy, reduced-order prediction, recalibration behavior, IRS-assisted throughput, controller design choices, end-to-end performance, and robustness in Figs.~\ref{fig:semi_blind_anchor_floor}-\ref{fig:robustness_mobility_load}. The differential repeated-data acquisition is evaluated independently of the resource controller using $\mathrm{NMSE}_{L,t}^{(k)}[\mathrm{dB}] = 10\log_{10} \frac{ \|\widehat{\mathbf L}_{t}^{(k)}-\mathbf L_t^{(k)}\|_F^2  }{\|\mathbf L_t^{(k)}\|_F^2 }$, reported in dB. Three references are considered: TRUE-H uses the true direct-channel matrix, O-CA-H denotes the practical acquisition based on the O-CA estimate, and Oracle-X assumes knowledge of the repeated user-data matrix. The semi-blind error covariance $\mathbf C_{\mathrm{SB}}^{(k)}$ is estimated from independent acquisition realizations and frozen before online evaluation.

\begin{figure*}[t!]
\centering
\includegraphics[width=\linewidth]{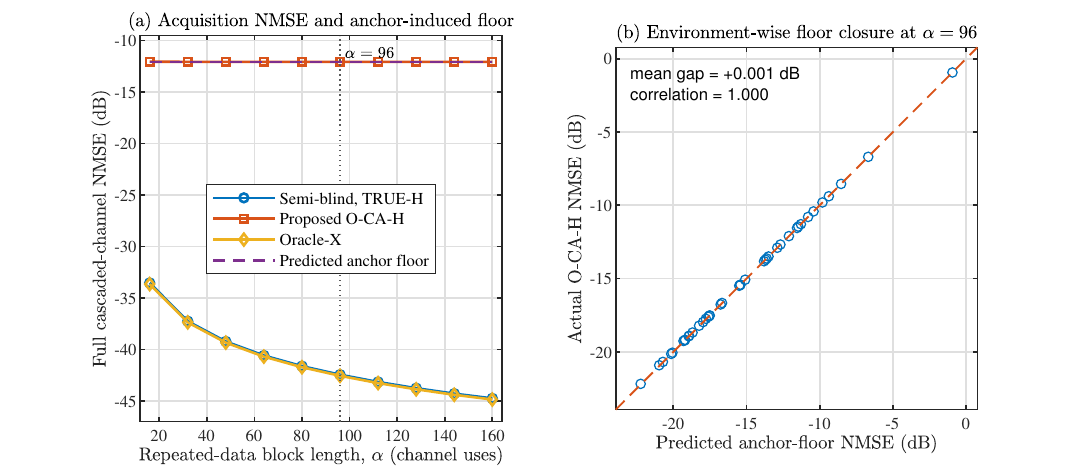}
\caption{Accuracy and direct-anchor-induced error floor of the differential semi-blind acquisition.}
\label{fig:semi_blind_anchor_floor}
\end{figure*}
Fig.~\ref{fig:semi_blind_anchor_floor} characterizes the finite-sample accuracy and anchor-induced error floor of the proposed semi-blind acquisition. In Fig.~\ref{fig:semi_blind_anchor_floor}(a), the TRUE-H and Oracle-X references remain nearly coincident, with the NMSE improving from approximately $-33.5$~dB at $\alpha=16$ to $-44.8$~dB at $\alpha=160$, confirming accurate differential DFT recovery when reliable multiuser separation is available. In contrast, the practical O-CA-H result remains near $-12.06$~dB throughout the sweep, showing that a longer repeated-data block cannot remove the error caused by an imperfect direct-channel anchor. At the nominal $\alpha=96$, the predicted anchor-only floor and simulated O-CA-H NMSE are $-12.065$ and $-12.064$~dB, respectively, with an environment-wise correlation of $1.000$, as shown in
Fig.~\ref{fig:semi_blind_anchor_floor}(b). The mismatch matrix contains $17.57\%$ mean off-diagonal energy, while its median and $90$th-percentile condition numbers are only $1.385$ and $1.969$. The observed floor is mainly caused by cross-user mixing due to direct-channel estimation error, rather than numerical ill-conditioning or insufficient repeated-data length. This acquisition error is therefore incorporated into $\mathbf C_{\mathrm{SB}}^{(k)}$ for subsequent uncertainty-aware control.

\begin{figure*}[t!]
\centering
\includegraphics[width=\linewidth]{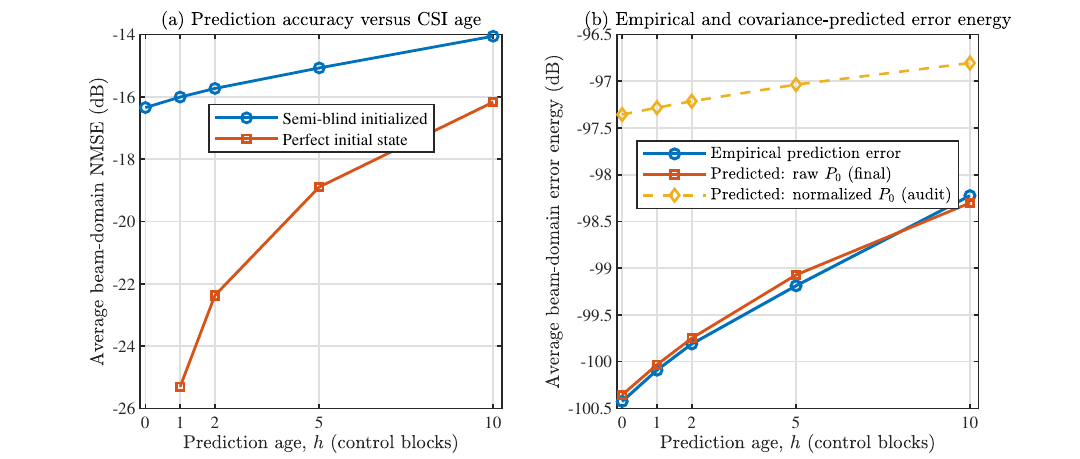}
\caption{Validation of the reduced-order beam-domain prediction and uncertainty model under the frozen structured physical channel.}
\label{fig:beam_domain_prediction_validation}
\end{figure*}
Fig.~\ref{fig:beam_domain_prediction_validation} validates the reduced-order beam-domain predictor initialized by the practical semi-blind acquisition. As shown in Fig.~\ref{fig:beam_domain_prediction_validation}(a), the average NMSE degrades only from $-16.344$~dB at $h=0$ to $-14.052$~dB at $h=10$. Although the perfect-initial-state reference is substantially more accurate at short prediction ages, its advantage decreases from about $9.30$~dB at $h=1$ to $2.12$~dB at $h=10$, indicating that temporal process uncertainty gradually becomes more important than the initial acquisition error. Fig.~\ref{fig:beam_domain_prediction_validation}(b) further shows that the raw covariance retained in the final controller closely follows the empirical prediction-error energy. For $h=0,1,2,5,$ and $10$, the empirical-to-predicted differences are $-0.070$, $-0.058$, $-0.061$, $-0.117$, and $0.078$~dB, respectively, remaining within $0.12$~dB of ideal aggregate calibration, whereas the normalized covariance is consistently more conservative. These results support the raw post-acquisition covariance and age-indexed reduced-order uncertainty model in the subsequent effective-channel prediction, recalibration, and reliability analysis.

\begin{figure*}[t!]
\centering
\includegraphics[width=\linewidth]{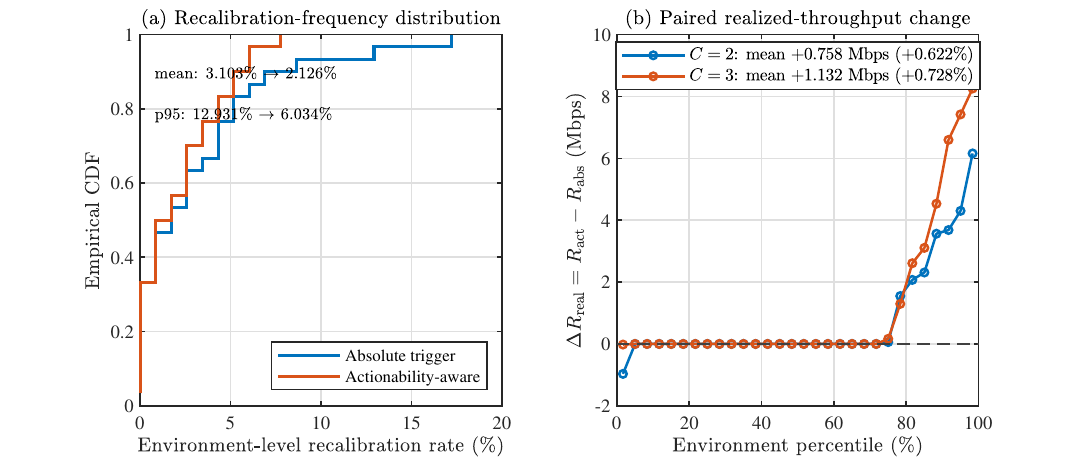}
\caption{Effect of the actionability-aware recalibration rule on the untouched paired holdout.}
\label{fig:actionability_holdout}
\end{figure*}
Consistent with Table~\ref{tab:actionability_holdout}, Fig.~\ref{fig:actionability_holdout} illustrates the distributional effect of the actionability-aware rule. Fig.~\ref{fig:actionability_holdout}(a) shows a $31.5\%$ reduction in the mean recalibration rate together with a clear reduction in upper-tail variability and repeated recalibration events.
Fig.~\ref{fig:actionability_holdout}(b) shows that this lower acquisition frequency does not reduce realized throughput: the average paired gains are $0.758$~Mbps ($0.622\%$) for $C=2$ and $1.132$~Mbps ($0.728\%$) for $C=3$. Moreover, $70\%$ of the environments exhibit identical throughput, while only $3.33\%$ show a negative paired change. Together with the fairness and minimum-rate results in Table~\ref{tab:actionability_holdout}, these findings show that actionability substantially reduces recalibration overhead while preserving the controller's rate-fairness behavior.

\begin{figure*}[t!]
\centering
\includegraphics[width=\linewidth]{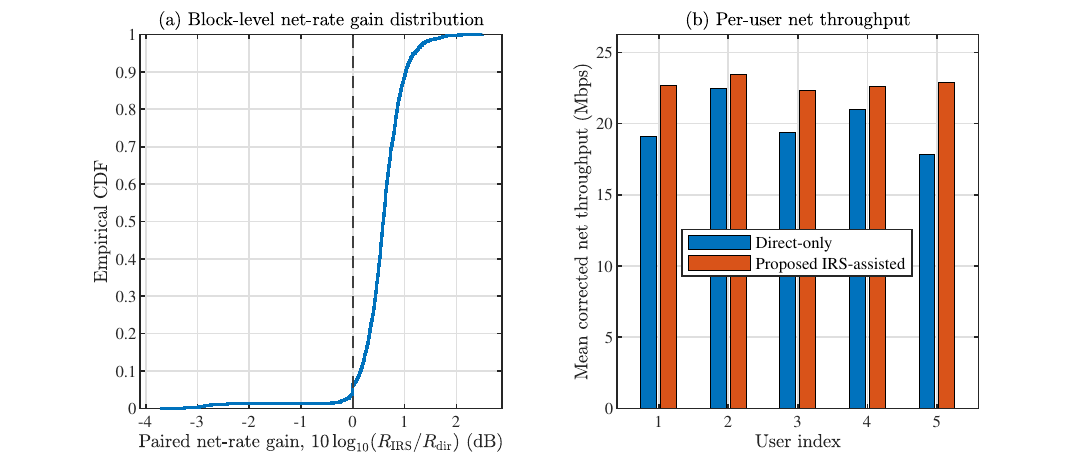}
\caption{Realized net-throughput gain of the proposed IRS-assisted physical-layer control relative to the direct-only reference.}
\label{fig:irs_direct_phy_gain}
\end{figure*}
Fig.~\ref{fig:irs_direct_phy_gain} compares the realized physical-layer performance of the proposed IRS-assisted control with direct-only transmission under identical channel realizations and block-level payload accounting. The mean realized net sum rate increases from $79.966$~Mbps to $91.234$~Mbps, corresponding to a $14.09\%$ gain after acquisition overhead. As shown in Fig.~\ref{fig:irs_direct_phy_gain}(a), $93.87\%$ of the evaluated blocks achieve a positive paired net-rate gain, with $10$th-, $50$th-, and $90$th-percentile gains of $0.154$, $0.587$, and $1.022$~dB, respectively. The correlation between the predicted and realized IRS-assisted net rates is $0.914$, supporting the uncertainty-aware predicted rate as the physical ranking metric of the structured controller. Fig.~\ref{fig:irs_direct_phy_gain}(b) further shows that all five users achieve higher mean net throughput, with individual gains ranging from $4.20\%$ to $28.20\%$. The pooled active-user $10$th-percentile rate increases from $9.830$~Mbps to $11.737$~Mbps, corresponding to an improvement of approximately $19.4\%$. Thus, the aggregate IRS gain is not concentrated on a small subset of users and remains favorable for lower-tail user performance.

\begin{figure*}[t!]
\centering
\includegraphics[width=\linewidth]{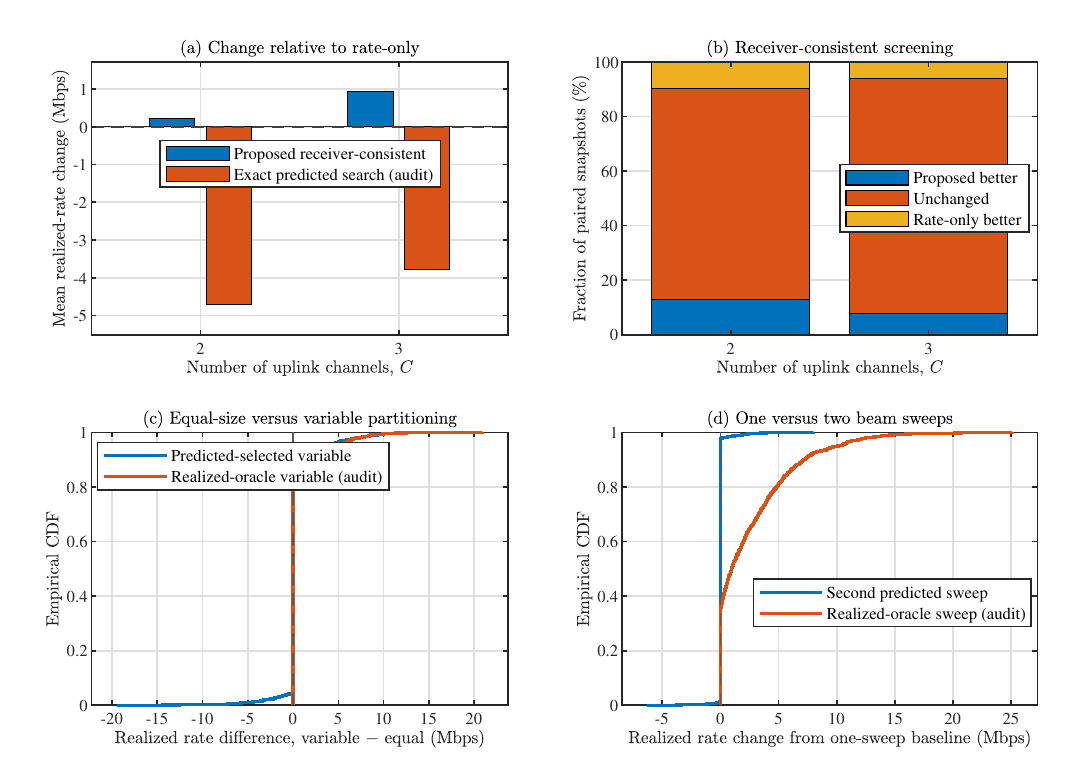}
\caption{Receiver-consistent channel-allocation validation and structural ablations of the final controller.}
\label{fig:receiver_structural_validation}
\end{figure*}
Fig.~\ref{fig:receiver_structural_validation} jointly validates the receiver-consistent channel-allocation rule and the structural choices retained in the final controller. In Fig.~\ref{fig:receiver_structural_validation}(a), the proposed correction increases the mean realized sum rate over the rate-only allocator by $0.228$~Mbps ($0.188\%$) for $C=2$ and $0.930$~Mbps ($0.605\%$) for $C=3$, whereas the exact predicted-rate search reduces the realized rate by $4.716$ and $3.789$~Mbps, respectively, despite maximizing the predicted objective. Fig.~\ref{fig:receiver_structural_validation}(b) confirms that the correction is selective: the rate-only assignment remains unchanged in $77.56\%$ and $86.35\%$ of the paired snapshots, while the proposed rule improves the realized rate more often than it reduces it for both channel counts. The ablations in Fig.~\ref{fig:receiver_structural_validation}(c) and (d) show that further structural complexity provides little repeatable gain. Variable partitioning yields only $0.283$~Mbps mean improvement over the equal-size design, with zero median gain and a positive gain in only $10.43\%$ of comparable epochs, while a second predicted beam sweep changes the beam vector in only $3.76\%$ of the slots and provides only $0.027$~Mbps mean gain.

The corresponding realized-oracle audits retain $0.668$ and $2.528$~Mbps of mean headroom, indicating that the remaining opportunity is mainly associated with prediction mismatch rather than insufficient partition or beam-search flexibility. These results support the receiver-consistent correction together with equal-size IRS candidates and single-sweep beam refinement in the final controller. The finest admissible granularity, $N_0=8$, is selected in $77.64\%$ of the steady blocks, with a similar granularity distribution in normal and recalibration blocks. A separate channel-load penalty provides no repeatable gain and is therefore assigned zero weight, supporting the two-term channel utility of Section~IX.

\begin{figure*}[t!]
\centering
\includegraphics[width=\linewidth]{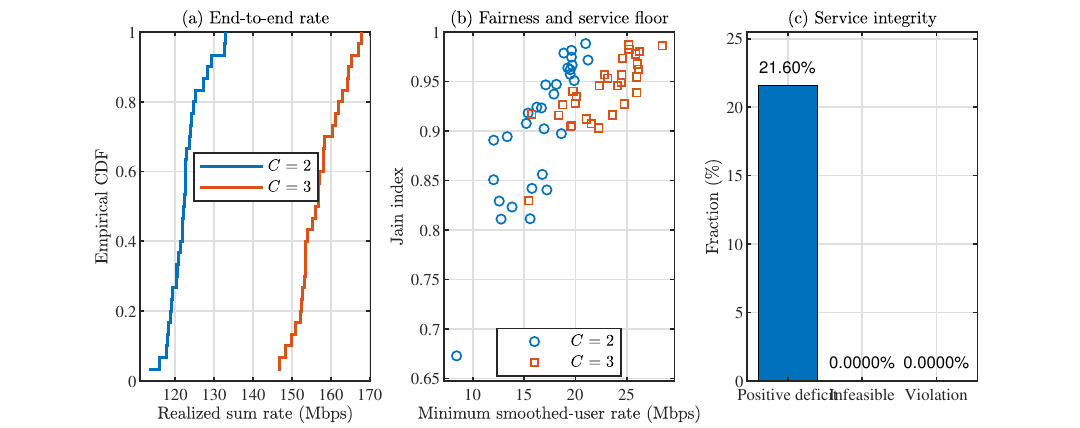}
\caption{End-to-end performance of the final deterministic controller on the untouched physical holdout.}
\label{fig:end_to_end_performance}
\end{figure*}
Fig.~\ref{fig:end_to_end_performance} evaluates the complete deterministic controller after fixing the acquisition, IRS-control, recalibration, and channel-allocation decisions. Across the $30$ untouched physical environments, the realized full-block sum rates are $122.492\pm4.466$~Mbps for $C=2$ and $156.667\pm5.858$~Mbps for $C=3$, with corresponding mean Jain indices of $0.9041$ and $0.9396$ and minimum smoothed-user rates of $16.680$ and $22.716$~Mbps, respectively. Thus, the aggregate throughput is obtained while maintaining favorable user-level fairness and minimum-rate performance. A positive service deficit occurs in $21.60\%$ of the active-user instances, confirming that the service-feasibility mechanism is actively exercised, while both the fraction of infeasible steady states and the post-action service-violation rate remain zero. The final controller therefore maintains feasible sliding-window service decisions together with stable throughput and fairness across previously unseen mobile environments.

\begin{figure*}[t!]
\centering
\includegraphics[width=\linewidth]{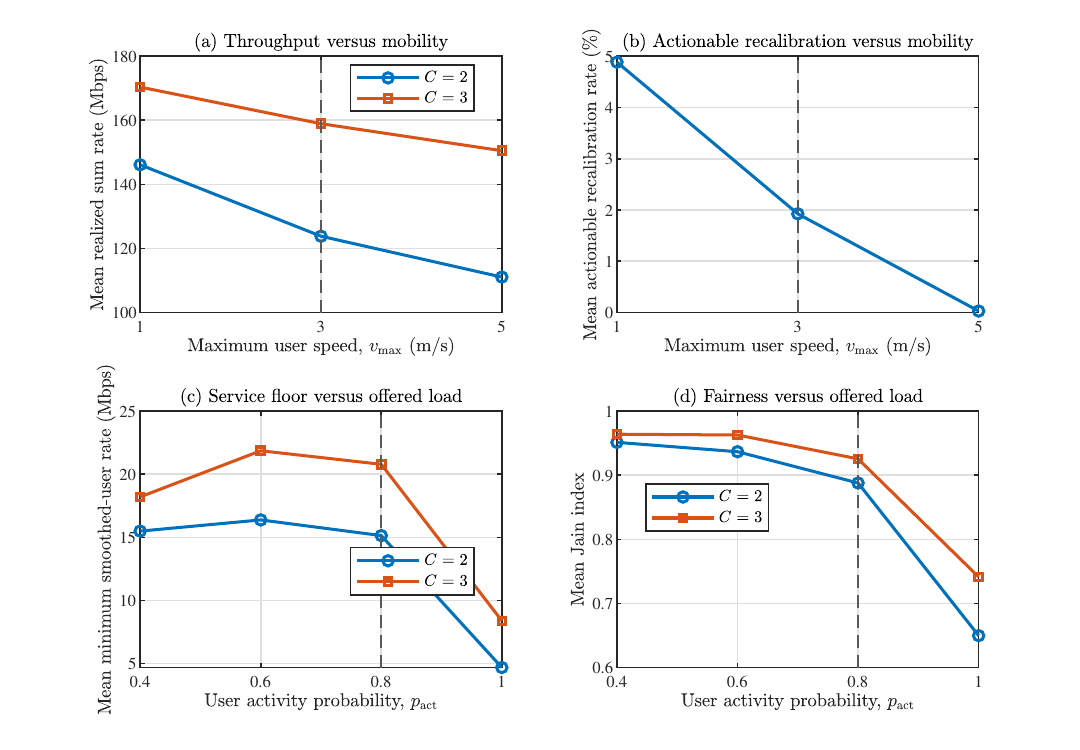}
\caption{Robustness and stress response of the final deterministic controller to mobility and offered load.}
\label{fig:robustness_mobility_load}
\end{figure*}
Fig.~\ref{fig:robustness_mobility_load} evaluates the frozen controller under varying mobility and offered load without parameter retuning. Increasing $v_{\max}$ from $1$ to $5$~m/s reduces the mean realized sum rate from $146.123$ to $111.056$~Mbps for $C=2$ and from $170.376$ to $150.506$~Mbps for $C=3$, reflecting the increased difficulty of channel
prediction. Meanwhile, the mean actionable recalibration rate decreases from $4.885\%$ to $0.029\%$, not because of improved tracking, but because the actionability condition increasingly rejects recalibrations whose predicted benefit does not justify their overhead. Under the load sweep, the minimum smoothed-user rate and Jain index remain favorable up to the nominal $p_{\mathrm{act}}=0.8$, reaching $15.138/20.791$~Mbps and $0.8880/0.9257$ for $C=2/3$, respectively, but decrease markedly under fully active traffic, identifying a saturation regime. No steady infeasible state or post-action service violation occurs at any tested mobility or load point. Thus, the controller preserves service feasibility across the tested stress conditions while exhibiting performance degradation over the relevant mobility and moderate-load range.

\section{Conclusion}
This paper addressed joint CSI maintenance and resource control in partitioned IRS-assisted mobile mMIMO IoT uplinks. O-CA tracks the direct links, while differential semi-blind acquisition refreshes cascaded links and finite-block calibration quantifies their uncertainty; between refreshes, only unit-beam means are propagated with age-dependent covariance lookup. A mixed-integer model captures throughput, fairness, outage, switching cost, and sliding-window IRS service, and is realized through actionable recalibration, feasible equal-size partitioning, residual-feasible ownership, one-sweep beam refinement, same-block payload accounting, and service-first receiver-consistent allocation with updated uncertainty-aware ZF rates. The analysis proves service-feasibility preservation, effective-channel prediction-error and distribution-free reliability bounds, and polynomial online complexity. Simulations reproduce the anchor-induced acquisition floor, validate the reduced-order uncertainty model, reduce mean recalibration by \(31.5\%\), and provide a \(14.09\%\) net-rate gain over direct-only transmission. Receiver-consistent allocation improves realized rate; ablations favor equal-size partitioning and single-sweep refinement, with no post-action service violations under nominal, mobility, or load stress.

Future work could evaluate practical hardware impairments and control delays, supported by validation using measured channels and experimental IRS platforms. Further work could also adapt the covariance-age tables under changing propagation conditions and extend to multicell systems with intercell interference and coordinated IRS control.

\end{document}